\documentclass[11pt,a4paper]{article}

\usepackage[utf8]{inputenc}
\usepackage[T1]{fontenc}
\usepackage{amsmath,amssymb,amsthm,amsfonts}
\usepackage{mathtools}
\usepackage{array}
\usepackage{enumitem}
\usepackage[hidelinks]{hyperref}
\usepackage[margin=1in]{geometry}
\usepackage{booktabs}
\usepackage{authblk}
\theoremstyle{plain}
\usepackage{color}
\newtheorem{theorem}{Theorem}[section]
\newtheorem{proposition}[theorem]{Proposition}
\newtheorem{lemma}[theorem]{Lemma}
\newtheorem{corollary}[theorem]{Corollary}

\theoremstyle{definition}
\newtheorem{definition}[theorem]{Definition}

\theoremstyle{remark}
\newtheorem{remark}[theorem]{Remark}
\newtheorem{example}[theorem]{Example}

\newcommand{\F}{\mathbb{F}}
\newcommand{\N}{\mathbb{N}}

\newcommand{\cC}{\mathcal{C}}
\newcommand{\cB}{\mathcal{B}}
\newcommand{\cM}{\mathcal{M}}
\newcommand{\cX}{\mathcal{X}}

\newcommand{\bn}{\mathbf{n}}
\newcommand{\bm}{\mathbf{m}}
\newcommand{\br}{\mathbf{r}}
\newcommand{\bV}{\mathbf{V}}
\newcommand{\bW}{\mathbf{W}}
\newcommand{\rk}{\operatorname{rk}}
\newcommand{\trk}{\operatorname{trk}}
\newcommand{\btrk}{\operatorname{btrk}}
\newcommand{\wt}{\operatorname{wt}}
\newcommand{\SR}{\mathrm{SR}}
\newcommand{\dSR}{d_{\mathrm{SR}}}
\newcommand{\GL}{\operatorname{GL}}
\newcommand{\diag}{\operatorname{diag}}
\newcommand{\supp}{\operatorname{supp}}
\newcommand{\ssone}{\mathrm{ss}_1}
\DeclareMathOperator{\Image}{Im}
\DeclareMathOperator{\Span}{span}
\newcommand{\Gri}{G_q}
\newcommand{\emb}{\varepsilon}

\title{Block Tensor Rank of Sum-Rank Metric Codes}

\author[1]{Huimin Lao}
\author[2]{Huy Pham}
\author[2]{Hoang Ta}
\author[3]{Van Khu Vu}

\affil[1]{\small Nanyang Technological University, Singapore}
\affil[2]{\small Hanoi University of Science and Technology, Vietnam}
\affil[3]{\small VinUniversity, Vietnam}
\date{\today}

\begin{document}
\maketitle
\begin{abstract}
    Sum-rank codes provide a generalized framework for Hamming and rank-metric codes, with codewords represented as tuples of matrices and weight given by the sum of the block ranks. In this paper, we introduce and study a block-tensor-rank invariant for sum-rank metric codes. To each code, we associate its \emph{block tensor rank}: the smallest number of block-simple tensors, namely rank-one matrices supported inside single blocks, whose linear span contains the code. In general, determining the block tensor rank of a sum-rank code is challenging. Our main structural result shows that the block tensor rank decomposes additively across the blocks of the code, thereby reducing its computation to a tensor-rank problem on each block projection. Consequently, we derive two complementary lower bounds on the block tensor rank, referred to as the \emph{projection-wise bound} and the \emph{coordinate-code bound}. Moreover, by combining the coordinate-code bound with the classical Singleton and Griesmer bounds for codes in the Hamming metric, we obtain explicit lower bounds, called the \emph{Singleton coordinate-code bound} and the \emph{Griesmer coordinate-code bound}, respectively. We further construct families of sum-rank codes whose block tensor ranks attain the Singleton or Griesmer coordinate-code bounds. These constructions are based on Hamming-metric codes achieving the corresponding classical bounds. Finally, we show that, in certain cases, the block tensor ranks of two known families of sum-rank codes in the literature do not attain the Singleton coordinate-code bound.
\end{abstract}


\section{Introduction}
\label{sec:introduction}

The sum-rank metric was introduced in~\cite{MartinezPenas2018} as a common framework for Hamming-metric and rank-metric codes. In this metric, a codeword is a tuple of matrices, partitioned into a fixed number of blocks, and its weight is the sum of the ranks of those blocks. Taking all blocks to have size \(1\times1\) recovers the Hamming metric, while taking a single block of arbitrary size recovers the rank metric. The sum-rank metric has since become an important object in coding theory, with applications to multishot network coding, space-time coding, distributed storage, and related network-coding models; see, for instance,~\cite{MartinezPenas2018,MartinezPenasKschischang2019,ByrneGluesingLuersse,SilvaKschischangKoetter2008,EtzionSilberstein2009,Ravagnani2016}.

The Singleton bound for the sum-rank metric generalizes both the classical Singleton bound for Hamming-metric codes and the Singleton bound for rank-metric codes. Codes attaining the sum-rank Singleton bound are called \emph{maximum sum-rank distance} (MSRD) codes. In the sum-rank setting, they play a role analogous to that of maximum-distance separable codes in the Hamming metric and maximum-rank-distance codes in the rank metric. Explicit MSRD families, including linearized Reed--Solomon codes, have been developed and studied extensively in~\cite{Delsarte1978,Gabidulin1985,MartinezPenas2018,MartinezPenasKschischang2019,Ravagnani2016}.

A parallel direction studies rank-metric codes through tensor decompositions. From this viewpoint, a rank-metric code is represented by a generator \(3\)-tensor, and the tensor rank of the code is the minimum number of rank-one matrices whose linear span contains it~\cite{BNRS19}. This tensor-rank invariant satisfies a Singleton-type lower bound: if \(\cC\) is a rank-metric code with tensor rank \(\trk(\cC)\) and minimum rank distance \(d_{\rm rk}(\cC)\), then
\[
\trk(\cC)\ge \dim(\cC)+d_{\rm rk}(\cC)-1.
\]
Codes meeting this bound are called \emph{minimum tensor rank} (MTR) codes. The tensor viewpoint connects rank-metric coding theory with classical questions on tensor decompositions, algebraic complexity, and linear block codes~\cite{Kruskal1977,BNRS19,Landsberg2012,BiniLotti1980}.

Since sum-rank codes generalize rank-metric codes, it is natural to ask whether the tensor-rank invariant of rank-metric codes admits a sum-rank analogue. The sum-rank metric is defined by summing the ranks of the individual blocks, so the natural rank-one building blocks should also respect the block decomposition. Thus, we restrict to rank-one matrices supported inside a single block, and call them block-simple tensors. We define the \emph{block tensor rank} of a sum-rank code as the smallest number of block-simple tensors whose \(\F_q\)-linear span contains the code. This invariant recovers the usual tensor rank in the rank-metric case and the minimum coordinate-support size in the Hamming case. Thus, the intermediate sum-rank regime is the setting in which this block-adapted tensor-rank invariant differs from both classical extremes. 

This invariant is also distinct from support-theoretic invariants in the sum-rank metric. Sum-rank supports, support spaces, generalized sum-rank weights, and optimal anticodes measure ambient support structures of codes~\cite{martinezpenas2019theory,camps2022optimal}. By contrast, block tensor rank measures the minimum number of rank-one directions inside the prescribed blocks needed to linearly span the code. In the Hamming specialization it recovers effective support size, whereas for genuine matrix blocks it refines support by resolving each block into rank-one directions.

Tensor methods have also appeared in the study of code equivalence problems for sum-rank codes. In particular, sum-rank code equivalence has been encoded in tensor-isomorphism terms, relating it to monomial tensor isomorphism and rank-metric code equivalence~\cite{d2024monomial}. That line of work uses tensor representations as a tool for complexity reductions and equivalence testing. Our focus is different: we do not study the equivalence of two given sum-rank codes. Instead, we introduce a block-simple spanning invariant of a single code and develop lower bounds, extremal notions, and constructions for this invariant. Thus, the present work should be viewed as a block-tensor-rank analogue of the tensor-rank theory of rank-metric codes, adapted to the block structure of the sum-rank metric.

\par 

\textbf{Contributions.} The paper makes the following contributions.

\begin{enumerate}[label=\textup{(\arabic*)}, leftmargin=*, itemsep=4pt]

\item \emph{A block-tensor-rank invariant for sum-rank codes.}
We introduce block-simple tensors and the associated block tensor rank. This gives a numerical invariant of a single sum-rank code, adapted to the block structure of the sum-rank metric. It specializes to the usual tensor rank for rank-metric codes and to minimum coordinate-support size for Hamming-metric codes.

\item \emph{A blockwise decomposition theorem.}
We prove that the block tensor rank of a sum-rank code decomposes as the sum of the tensor ranks of its block projections. Thus, the computation of the block tensor rank separates into rank-metric tensor-rank problems on the individual blocks.


\item \emph{Lower bounds on block tensor rank}.
We prove two complementary lower bounds for block tensor rank of sum-rank codes. The first one is the projection-wise bound following from the blockwise decomposition. 
The second one is the coordinate-code bound, which says that for a sum-rank code \(\mathcal C\) of dimension \(k\) and minimum sum-rank
distance \(d\), $\operatorname{btrk}(\mathcal C)\ge N_q(k,d),$ where \(N_q(k,d)\) denotes the minimum length of a linear code of dimension \(k\) and minimum Hamming distance \(d\). The codes attaining this bound are called \emph{block-tensor-rank-extremal codes}.
By classical Singleton bound and Griesmer bound, we further obtain \begin{align*}
    \operatorname{btrk}(\mathcal{C}) & \geq k+d-1, &  (\text{Singleton coordinate-code bound}), \\
    \operatorname{btrk}(\mathcal{C}) &\geq G_{q}(k,d):=\sum_{j=0}^{k-1}\left\lceil\frac{d}{q^j}\right\rceil, & (\text{Griesmer coordinate-code bound}).
\end{align*}
A code attaining the bound \(N_q(k,d)\) is called \emph{block-tensor-rank-extremal}. A code attaining the Singleton-level bound
\(k+d-1\) is called a \emph{block tensor rank minimum} code, or a \emph{BTR} code. Since \(N_q(k,d)\ge k+d-1\), BTR is a stronger Singleton-level condition and can occur only when \(N_q(k,d)=k+d-1\).

\item \emph{Constructions of block-tensor-rank-extremal codes}.
We present a block-diagonal lifting method and identify three explicit parameter regimes in which it yields BTR sum-rank codes.
The sum-rank codewords are obtained by partitioning the codewords of MDS codes into several small blocks and then transforming each of them to matrices by block-diagonal maps. Assuming the existence of an MDS \([R,k,d]_q\) backbone code with \(R=k+d-1\),
our method yields BTR sum-rank codes in $\cM(\bn,\bm)$ (see (\ref{eq: mat_space})) in several parameter regimes, summarized in Table~\ref{tab:btr-regimes}.

\begin{table}[ht]
\centering
\renewcommand{\arraystretch}{1.25}
\begin{tabular}{|c|c|c|}
\hline
\textbf{Regime} & \textbf{Hypothesis} & \textbf{Result} \\
\hline
Fat block
&
$\min(n_{i_0},m_{i_0})\ge R$, \text{for an} $i_{0} \in [t]$
&
Proposition~\ref{prop:fat-block}
\\
\hline
MDS backbone, uniform
&
$n_i=n\ge d,\; m_i=m\ \text{for all }i,\; n+m\ge r+d$
&
Corollary~\ref{cor:uniform-ngeqd}
\\
\hline
Pooling regime, small $k$
&
$n_i=n,\; m_i=m, \; n<d\le t n,\; r\le\min(n,m)$
&
Corollary~\ref{cor:uniform-pooling}
\\
\hline
\end{tabular}
\caption[Explicit BTR regimes from block-diagonal lifting]{Explicit BTR regimes obtained from block-diagonal lifting with an MDS \([R,k,d]_q\) backbone code, where \(R=k+d-1\). In the last two regimes, the number of blocks $t=R/r$.}
\label{tab:btr-regimes}
\end{table}

In addition, we provide an explicit construction to obtain BTR sum-rank codes. Furthermore, using the same lifting method with linear codes attaining the classical Griesmer bound, such as simplex codes and MacDonald codes, we obtain block-tensor-rank-extremal codes whose block tensor rank attains the Griesmer coordinate-code bound.

\item \emph{Comparison with MSRD codes and BTR codes.} We compare block-tensor-rank optimality with MSRD optimality. We show that the BTR codes obtained in the uniform MDS-backbone lifting regime considered in this paper cannot simultaneously attain the uniform MSRD dimension. We also derive a projection-gap criterion that quantifies how the tensor ranks of the block projections force a code away from the BTR threshold. Applying this criterion to existing sum-rank constructions gives both close-to-BTR and large-gap examples.

\end{enumerate}

\textbf{Organization.}
Section~\ref{sec:preliminaries} fixes the notation for sum-rank metric codes,
recalls the relevant Hamming-metric and tensor-rank preliminaries, and
introduces block-simple tensors and block tensor rank. Section~\ref{sec:blockwise-projection}
proves the blockwise decomposition identity and derives the projection-wise
lower bound, together with elementary upper bounds and invariance properties.
Section~\ref{sec:coordinate-bounds} develops the coordinate-code bound, derives
the Singleton and Griesmer coordinate-code bounds, and introduces the associated
extremal notions. Section~\ref{sec:constructions} presents the block-diagonal
lifting method and constructs BTR and block-tensor-rank-extremal families.
Section~\ref{sec:msrd-comparison} compares BTR and MSRD optimality, applies the
framework to existing constructions, and discusses encoding complexity.
Section~\ref{sec:conclusion} concludes with future directions.

\section{Preliminaries}
\label{sec:preliminaries}

\subsection{Sum-rank metric codes and notation}

Throughout the paper, $q$ denotes a prime power and $t\ge 1$ a positive integer (the number of blocks). We fix two tuples of positive integers
\[
\bn=(n_1,\dots,n_t),\qquad \bm=(m_1,\dots,m_t),
\]
recording the row and column sizes of the blocks. We set $N:=\sum_{i=1}^{t}n_i$ and $M:=\sum_{i=1}^{t}m_i$; for $\ell\ge 1$ we write $[\ell]:=\{1,\dots,\ell\}$. All vector spaces, dimensions and linear maps are taken over $\F_q$ unless stated otherwise. The ambient space of the sum-rank metric is the direct sum
\begin{align}
\label{eq: mat_space}
\cM(\bn,\bm) \coloneqq \bigoplus_{i=1}^{t}\F_q^{n_i\times m_i},
\end{align}
which is an $\F_q$-vector space of dimension $\sum_{i=1}^{t}n_im_i$. We write an element as $X=(X_1,\dots,X_t)$, with $X_i\in\F_q^{n_i\times m_i}$ the $i$-th block.

\begin{definition}
\label{def:sr-weight}
The \emph{sum-rank weight} of $X=(X_1,\dots,X_t)\in\cM(\bn,\bm)$ is $\wt_{\SR}(X)\coloneqq \sum_{i=1}^{t}\rk(X_i)$
and the \emph{sum-rank distance} is $\dSR(X,Y) \coloneqq \wt_{\SR}(X-Y)$.
\end{definition}

\begin{example}
\label{ex:sr-weight}
Let $q=2$, $t=2$, $\bn=(2,2)$, $\bm=(3,2)$, and consider
\[
X_1=\begin{pmatrix}1&0&1\\0&1&1\end{pmatrix},\qquad
X_2=\begin{pmatrix}1&1\\1&1\end{pmatrix}.
\]
The rows of $X_1$ are linearly independent over $\F_2$, so $\rk(X_1)=2$, while $X_2$ has all entries equal and so $\rk(X_2)=1$. Therefore $\wt_{\SR}((X_1,X_2))=3$.
\end{example}

\begin{definition}
\label{def:sr-code}
A \emph{sum-rank metric code} is an $\F_q$-linear subspace $\cC\subseteq\cM(\bn,\bm)$. If $\cC\ne\{0\}$, its \emph{minimum sum-rank distance} is
\[
\dSR(\cC)\;:=\;\min\{\wt_{\SR}(X):X\in\cC,\ X\ne 0\}.
\]
We say that $\cC$ is an $[\bn\times\bm,k,d]_q$ sum-rank code if $\dim_{\F_q}(\cC)=k$ and $\dSR(\cC)=d$.
\end{definition}

\begin{remark}
\label{rem:extremes}
If $n_i=m_i=1$ for all $i$, then $\cM(\bn,\bm)=\F_q^t$ and $\wt_{\SR}$ is the Hamming weight. If $t=1$, then $\cM(\bn,\bm)=\F_q^{n_1\times m_1}$ and $\wt_{\SR}$ is the rank weight.
\end{remark}

We record the sum-rank Singleton bound for reference; the proof in the general non-uniform setting is in~\cite{byrne2021fundamental}, and the uniform specialization is due to~\cite{MartinezPenas2018}.

\begin{theorem}[Sum-rank Singleton bound]
\label{thm:singleton}
Let $\cC$ be a nonzero $[\bn\times\bm,k,d]_q$ sum-rank code. Reorder the blocks so that $(\min\{n_i,m_i\})_{i=1}^{t}$ is non-increasing, and let $j$ and $\delta$ be the unique integers satisfying
\[
0\le\delta<\min\{n_j,m_j\},\qquad d-1=\sum_{i=1}^{j-1}\min\{n_i,m_i\}+\delta.
\]
Then
\[
k\;\le\;\sum_{i=j}^{t}\max\{n_i,m_i\}\min\{n_i,m_i\}\;-\;\delta\max\{n_j,m_j\}.
\]
In the uniform case $n_i=n\le m_i=m$ for all $i$, the bound reduces to $k\le m(tn-d+1)$. A code attaining the bound is called an \emph{MSRD code}.
\end{theorem}
\subsection{Classical Hamming-metric bounds and tensor rank}

\textbf{Notation conventions.}
For integers \(1\le a\le b\), we write $[a,b]:=\{a,a+1,\ldots,b\}.$ For a vector \(c\in\mathbb F_q^R\), we denote by
\[
\mathrm{wt}_H(c):=\#\{r\in[R]:c_r\ne 0\}
\]
its Hamming weight. If \(D\subseteq \mathbb F_q^R\) is a nonzero linear block
code, then
\[
d_H(D):=\min\{\mathrm{wt}_H(c):c\in D,\ c\ne 0\} 
\]
denotes its minimum Hamming distance. Throughout the paper, \(N_q(k,d)\)
denotes the minimum length \(N\) of a linear \([N,k,\ge d]_q\) block code over
\(\mathbb F_q\).

We recall the two classical Hamming-metric lower bounds on \(N_q(k,d)\) that
will be used later. The Singleton bound gives
\[
N_q(k,d)\ge k+d-1.
\]
For positive integers \(k\) and \(d\), define the Griesmer quantity
\[
\Gri(k,d):=\sum_{j=0}^{k-1}\left\lceil\frac{d}{q^j}\right\rceil.
\]
The Griesmer bound gives
\[
N_q(k,d)\ge \Gri(k,d).
\]

\textbf{Tensor-rank notation.}
We briefly recall the notation for $3$-tensors used in~\cite{BNRS19}. Let $\F$ be a field and $k,n,m\in\N$. A $3$-tensor $X\in\F^{k\times n\times m}$ admits a decomposition as a sum of \emph{simple tensors}
\[
X=\sum_{r=1}^{R}u_r\otimes v_r\otimes w_r,\qquad u_r\in\F^k,\ v_r\in\F^n,\ w_r\in\F^m,
\]
and the minimum such $R$ is the \emph{tensor rank} $\trk(X)$. For a matrix $A\in\F^{s\times k}$, the operation
\[
m_1(A,X)\;:=\;\sum_{r}(Au_r)\otimes v_r\otimes w_r
\]
is independent of the chosen decomposition and satisfies $m_1(A,m_1(B,X))=m_1(AB,X)$. Writing $\{e_j\}_{j=1}^{k}$ for the standard basis of $\F^k$, the \emph{first slice space} of $X$ is
\[
\ssone(X)\;:=\;\bigl\langle m_1(e_j,X):j\in[k]\bigr\rangle\subseteq\F^{n\times m}.
\]
We refer to~\cite[\S 3]{BNRS19} for a detailed treatment. When $\cC\subseteq\F^{n\times m}$ is a rank-metric code, its \emph{tensor rank} $\trk(\cC)$ is by definition the minimum number of rank-one matrices in $\F^{n\times m}$ whose $\F$-linear span contains $\cC$; equivalently~\cite[Definition~4.3]{BNRS19}, $\trk(\cC)$ is the minimum tensor rank of a generator tensor for $\cC$.

\subsection{Block-simple tensors}

The first new ingredient of the framework is the class of elementary tensors
that respect the block partition. For each $i\in[t]$, write
\[
\emb_i:\F_q^{n_i\times m_i}\hookrightarrow\cM(\bn,\bm)
\]
for the canonical inclusion into the $i$-th coordinate (with zeros on all other blocks).

\begin{definition}
\label{def:blocksimple}
An element \(B\in\cM(\bn,\bm)\) is \emph{block-simple} if there exist \(i\in[t]\), nonzero vectors \(v\in\F_q^{n_i}\) and \(w\in\F_q^{m_i}\)
such that
\[
B=\emb_i(vw^\top).
\]
In this case we say that \(B\) is \emph{supported on block~\(i\)}. We write \(\cB(\bn,\bm)\) for the set of all block-simple tensors.
\end{definition}

Under the block-diagonal embedding $\cM(\bn,\bm)\hookrightarrow\F_q^{N\times M}$, a block-simple tensor $\emb_i(vw^\top)$ becomes a rank-one matrix $\widetilde{v}\widetilde{w}^\top$ where $\widetilde{v}\in\F_q^N$ is $v$ padded with zeros outside the $i$-th row-block, and similarly for $\widetilde{w}$.

\begin{lemma}
\label{lem:blocksimple-span}
Every element of $\cM(\bn,\bm)$ is an $\F_q$-linear combination of block-simple tensors.
\end{lemma}

\begin{proof}
Let $X=(X_1,\dots,X_t)\in\cM(\bn,\bm)$, so that $X=\sum_{i=1}^{t}\emb_i(X_i)$. For each $i$, the matrix $X_i\in\F_q^{n_i\times m_i}$ is a sum of $\rk(X_i)\le\min(n_i,m_i)$ rank-one matrices over $\F_q$, say $X_i=\sum_{s}v_{i,s}w_{i,s}^\top$. Applying $\emb_i$,
\[
\emb_i(X_i)=\sum_{s}\emb_i(v_{i,s}w_{i,s}^\top),
\]
and each summand is block-simple. Summing over $i$ writes $X$ as a linear combination of block-simple tensors.
\end{proof}

\begin{example}
\label{ex:blocksimple}
Let $q=2$, $t=2$, $\bn=\bm=(2,2)$. Take $v=(1,0)^\top$, $w=(1,1)^\top$, so that
\[
\emb_1(vw^\top)=\biggl(\begin{pmatrix}1&1\\0&0\end{pmatrix},\,\begin{pmatrix}0&0\\0&0\end{pmatrix}\biggr)
\]
is block-simple of sum-rank weight $1$. By contrast,
\[
X=\biggl(\begin{pmatrix}1&0\\0&0\end{pmatrix},\,\begin{pmatrix}1&0\\0&0\end{pmatrix}\biggr)
\]
has sum-rank weight $2$ but is \emph{not} block-simple, even though its block-diagonal embedding into $\F_2^{4\times 4}$ has matrix rank $2$. We can however write $X=\emb_1(e_1e_1^\top)+\emb_2(e_1e_1^\top)$, expressing it as a sum of two block-simple tensors.
\end{example}

\subsection{Generator tuples}

In order to talk about tensor decompositions of a code, we attach to each code a tuple of generator tensors, one for each block.

\begin{definition}
\label{def:gentuple}
Let $\cC\subseteq\cM(\bn,\bm)$ be a sum-rank code of dimension $k$. A \emph{generator tuple} for $\cC$ is a tuple
\[
\cX=(X^{(1)},\dots,X^{(t)}),\qquad X^{(i)}\in\F_q^{k\times n_i\times m_i},
\]
such that
\[
\cC=\Bigl\{\bigl(m_1(a,X^{(1)}),\dots,m_1(a,X^{(t)})\bigr):a\in\F_q^k\Bigr\}.
\]
Equivalently, the map
\[
E_\cX:\F_q^k\to\cM(\bn,\bm),\qquad a\longmapsto\bigl(m_1(a,X^{(i)})\bigr)_{i=1}^{t},
\]
is an $\F_q$-linear monomorphism with image $\cC$.
\end{definition}

\begin{lemma}
\label{lem:gentuple-existence}
Every sum-rank code $\cC$ of dimension $k$ admits a generator tuple. Moreover, any two generator tuples for $\cC$ are related by a basis change: if $\cX$ and $\cX'$ both generate $\cC$, then there exists $P\in\GL(k,q)$ such that $X'^{(i)}=m_1(P,X^{(i)})$ for every $i\in[t]$.
\end{lemma}

\begin{proof}
\emph{Existence.} Choose a basis $A^{(1)},\dots,A^{(k)}$ of $\cC$ and write each element in block form as $A^{(j)}=(A^{(j)}_1,\dots,A^{(j)}_t)$ with $A^{(j)}_i\in\F_q^{n_i\times m_i}$. For each $i\in[t]$, define
\[
X^{(i)}:=\sum_{j=1}^{k}e_j\otimes A^{(j)}_i\in\F_q^{k\times n_i\times m_i}.
\]
Then for any $a=(a_1,\dots,a_k)\in\F_q^k$,
\[
m_1(a,X^{(i)})=\sum_{j=1}^{k}(a\cdot e_j)A^{(j)}_i=\sum_{j=1}^{k}a_jA^{(j)}_i,
\]
so $E_\cX(a)=\sum_{j}a_jA^{(j)}\in\cC$. Since $\{A^{(j)}\}$ is a basis, $E_\cX$ is injective with image $\cC$.

\emph{Basis-change relation.} The maps $E_\cX,E_{\cX'}\colon\F_q^k\to\cM(\bn,\bm)$ are injective $\F_q$-linear maps with the same image $\cC$. Hence there exists $P\in\GL(k,q)$ such that $E_{\cX'}(a)=E_\cX(aP)$ for all $a\in\F_q^k$, viewing $a$ as a row vector. Block by block,
\[
m_1(a,X'^{(i)})=m_1(aP,X^{(i)})=m_1(a,m_1(P,X^{(i)})),
\]
where we used $m_1(aP,X^{(i)})=m_1(a,m_1(P,X^{(i)}))$. Since this holds for all $a$, we conclude $X'^{(i)}=m_1(P,X^{(i)})$.
\end{proof}
We next record how the slice space of each component tensor in a generator
tuple is related to the corresponding block projection of the code. Recall
that a tensor \(X\in\mathbb F_q^{k\times n\times m}\) is called 
\emph{1-nondegenerate} if $\dim_{\mathbb F_q}\operatorname{ss}_1(X)=k.$ Equivalently, its first-mode slices are linearly independent.

\begin{lemma}
\label{lem:slice-proj}
Let $\cX$ be a generator tuple for $\cC$ and let $\pi_i\colon\cM(\bn,\bm)\to\F_q^{n_i\times m_i}$ denote projection onto the $i$-th block. Then $\ssone(X^{(i)})=\pi_i(\cC)$. In particular, $X^{(i)}$ is $1$-nondegenerate if and only if $\pi_i|_\cC$ is injective.
\end{lemma}

\begin{proof}
For any $a\in\F_q^k$, $m_1(a,X^{(i)})=\pi_i(E_\cX(a))$. Hence
\[
\ssone(X^{(i)})=\langle m_1(e_j,X^{(i)}):j\in[k]\rangle=\langle\pi_i(E_\cX(e_j)):j\in[k]\rangle=\pi_i(\cC),
\]
since $\{E_\cX(e_j)\}_{j=1}^{k}$ is a basis of $\cC$. The second statement follows because $\pi_i|_\cC$ is injective if and only if $\dim\pi_i(\cC)=k$.
\end{proof}

\begin{example}
\label{ex:gentuple}
Let $q=2$, $t=2$, $\bn=\bm=(2,2)$. Let $I_k$ be the $k \times k$ identity matrix and define
\[
A^{(1)}=(I_2,0),\qquad A^{(2)}=(0,I_2).
\]
A generator tuple for $\cC=\langle A^{(1)},A^{(2)}\rangle$ is $\cX=(X^{(1)},X^{(2)})$ with $X^{(1)}=e_1\otimes I_2$ and $X^{(2)}=e_2\otimes I_2$, so that $E_\cX(a_1,a_2)=(a_1I_2,a_2I_2)$. Each nonzero codeword has sum-rank weight at least $2$, so $\dSR(\cC)=2$.
\end{example}

\subsection{Block tensor rank}

\begin{definition}
\label{def:btrk}
Let \(\cX=(X^{(1)},\dots,X^{(t)})\), where
\(X^{(i)}\in\F_q^{k\times n_i\times m_i}\). The \emph{block tensor rank} of
\(\cX\), denoted \(\btrk(\cX)\), is the smallest integer \(R\ge0\) for which there
exist a function \(\sigma:[R]\to[t]\), vectors \(u_r\in\F_q^k\), and vectors
\(v_r\in\F_q^{n_{\sigma(r)}}\), \(w_r\in\F_q^{m_{\sigma(r)}}\), for \(r\in[R]\),
such that
\[
X^{(i)}
=
\sum_{r\in\sigma^{-1}(i)}
u_r\otimes v_r\otimes w_r
\qquad
\text{for every } i\in[t].
\]
\end{definition}

The block tensor rank is a numerical invariant of a single sum-rank code. It should be distinguished from tensor-isomorphism approaches to sum-rank
code equivalence, where one studies whether two codes lie in the same orbit under the relevant isometry group, as in~\cite{d2024monomial}.
Here, we instead measure the minimum number of sum-rank weight-one tensors needed to span a given code. The next proposition characterises the block tensor rank in three equivalent ways. It is the sum-rank analogue of~\cite[Proposition~3.4]{BNRS19}.

\begin{proposition}
\label{prop:btrk-char}
Let $\cX=(X^{(1)},\dots,X^{(t)})$ be a generator tuple for $\cC=\Image(E_\cX)$, and let $R\ge 0$. The following statements are equivalent.
\begin{enumerate}[label=\textup{(\arabic*)}]
\item $\btrk(\cX)\le R$.
\item There exist block-simple tensors $B_1,\dots,B_R\in\cB(\bn,\bm)$ such that $\cC\subseteq\langle B_1,\dots,B_R\rangle$.
\item There exist a function $\sigma\colon[R]\to[t]$, matrices
\[
V\in\F_q^{N\times R},\qquad W\in\F_q^{M\times R},
\]
whose $r$-th columns are supported on the $\sigma(r)$-th row-block and column-block respectively, together with diagonal matrices $D_1,\dots,D_k\in\F_q^{R\times R}$, such that
\[
\ssone(X^{(i)})=\langle\pi_i(VD_jW^\top):j\in[k]\rangle
\]
for every $i\in[t]$, where $\pi_i$ denotes projection to the $i$-th block under $\cM(\bn,\bm)\hookrightarrow\F_q^{N\times M}$.
\end{enumerate}
\end{proposition}

\begin{proof}
\emph{(1)$\Rightarrow$(2).} Suppose $\cX$ admits a decomposition of length $R$ as in Definition~\ref{def:btrk}, with data $(\sigma,u_r,v_r,w_r)$. Set $B_r:=\emb_{\sigma(r)}(v_rw_r^\top)$, which is block-simple. For $a\in\F_q^k$,
\begin{align*}
E_\cX(a)
&=\sum_{i=1}^{t}\emb_i(m_1(a,X^{(i)}))\\
&=\sum_{i=1}^{t}\emb_i\Bigl(\sum_{r\in\sigma^{-1}(i)}(a\cdot u_r)v_rw_r^\top\Bigr)\\
&=\sum_{r=1}^{R}(a\cdot u_r)\,\emb_{\sigma(r)}(v_rw_r^\top)
=\sum_{r=1}^{R}(a\cdot u_r)B_r,
\end{align*}
so every codeword is a linear combination of the $B_r$.

\emph{(2)$\Rightarrow$(1).} Suppose $\cC\subseteq\langle B_1,\dots,B_R\rangle$ where $B_r=\emb_{\sigma(r)}(v_rw_r^\top)$. For each $j\in[k]$, expand $E_\cX(e_j)=\sum_{r=1}^{R}\lambda_{j,r}B_r$ and set $u_r:=(\lambda_{1,r},\dots,\lambda_{k,r})^\top\in\F_q^k$. Projecting onto block $i$ gives
\[
m_1(e_j,X^{(i)})=\pi_i(E_\cX(e_j))=\sum_{r\in\sigma^{-1}(i)}\lambda_{j,r}\,v_rw_r^\top,
\]
which is the $j$-th slice of $\sum_{r\in\sigma^{-1}(i)}u_r\otimes v_r\otimes w_r$. Since slices determine a $3$-tensor uniquely, $X^{(i)}=\sum_{r\in\sigma^{-1}(i)}u_r\otimes v_r\otimes w_r$, a decomposition of length $R$.

\emph{(1)$\Rightarrow$(3).} Given a length-$R$ decomposition $(\sigma,u_r,v_r,w_r)$, define $V\in\F_q^{N\times R}$ by placing $v_r$ in the $\sigma(r)$-th row-block of the $r$-th column (zero elsewhere), and similarly for $W$. For each $j\in[k]$ let $D_j:=\diag(u_{j,r}:r\in[R])$, where $u_{j,r}$ is the $j$-th coordinate of $u_r$. Then
\[
VD_jW^\top=\sum_{r=1}^{R}u_{j,r}\,\widetilde{v}_r\widetilde{w}_r^\top,
\]
where the padded matrix $\widetilde{v}_r\widetilde{w}_r^\top$ has its only nonzero block at the $(\sigma(r),\sigma(r))$ position, equal to $v_rw_r^\top$. Thus
\[
\pi_i(VD_jW^\top)=\sum_{r\in\sigma^{-1}(i)}u_{j,r}\,v_rw_r^\top=m_1(e_j,X^{(i)}),
\]
and taking the span over $j$ gives the desired identity.

\emph{(3)$\Rightarrow$(2).} For each $r\in[R]$, let $v_r$ be the nonzero portion of the $r$-th column of $V$ in row-block $\sigma(r)$, and $w_r$ the nonzero portion of the $r$-th column of $W$ in column-block $\sigma(r)$. Set $B_r:=\emb_{\sigma(r)}(v_rw_r^\top)$. By the column-support hypothesis, the $i$-th diagonal block of $VD_jW^\top$ lies in $\langle v_rw_r^\top:r\in\sigma^{-1}(i)\rangle$. Combined with the identity in (3),
\[
\pi_i(\cC)=\ssone(X^{(i)})\subseteq\langle\pi_i(B_r):r\in\sigma^{-1}(i)\rangle.
\]
For any $X=(X_1,\dots,X_t)\in\cC$, each $X_i$ is therefore a linear combination of $\{\pi_i(B_r):r\in\sigma^{-1}(i)\}$, and embedding these back into the corresponding blocks shows $X\in\langle B_1,\dots,B_R\rangle$. This gives (2), and we already showed (2)$\Rightarrow$(1).
\end{proof}

\begin{proposition}
\label{prop:btrk-welldefined}
If $\cX$ and $\cX'$ are two generator tuples for the same sum-rank code $\cC$, then $\btrk(\cX)=\btrk(\cX')$.
\end{proposition}

\begin{proof}
By Lemma~\ref{lem:gentuple-existence}, there exists $P\in\GL(k,q)$ with $X'^{(i)}=m_1(P,X^{(i)})$ for every $i$. Suppose $\cX$ admits a length-$R$ decomposition $X^{(i)}=\sum_{r\in\sigma^{-1}(i)}u_r\otimes v_r\otimes w_r$. Applying $m_1(P,\cdot)$,
\[
X'^{(i)}=\sum_{r\in\sigma^{-1}(i)}(Pu_r)\otimes v_r\otimes w_r,
\]
which is a length-$R$ decomposition of $\cX'$ using the same block-simple tensors $\emb_{\sigma(r)}(v_rw_r^\top)$. Hence $\btrk(\cX')\le\btrk(\cX)$, and the reverse inequality follows by symmetry.
\end{proof}

\begin{definition}
\label{def:btrk-code}
The \emph{block tensor rank} of a nonzero sum-rank code $\cC$ is $\btrk(\cC):=\btrk(\cX)$, where $\cX$ is any generator tuple of $\cC$; equivalently, $\btrk(\cC)$ is the smallest integer $R$ such that $\cC$ is contained in the $\F_q$-linear span of $R$ block-simple tensors.
\end{definition}

\section{Blockwise Decomposition and Projection-wise Bounds}
\label{sec:blockwise-projection}

\subsection{The blockwise decomposition identity}
\label{subsec:blockwise}

We now state and prove the central structural identity. It reduces the computation of $\btrk(\cC)$ to a tensor-rank computation for each block projection separately, and is the basis for almost every later result.

\begin{proposition}[Blockwise decomposition]
\label{prop:blockwise-decomposition}
Let $\cC\subseteq\cM(\bn,\bm)$ be a nonzero sum-rank code, and set $\cC_i:=\pi_i(\cC)\subseteq\F_q^{n_i\times m_i}$. Then
\[
\btrk(\cC)\;=\;\sum_{i=1}^{t}\trk(\cC_i),
\]
where $\trk(\cC_i)$ is the tensor rank of $\cC_i$ as a rank-metric code in the sense of~\cite{BNRS19}, with the convention $\trk(\{0\})=0$.
\end{proposition}

\begin{proof}
\emph{Lower bound: $\btrk(\cC)\ge\sum_i\trk(\cC_i)$.}
Let $R=\btrk(\cC)$ and let $B_1,\dots,B_R\in\cB(\bn,\bm)$ be block-simple tensors whose span contains $\cC$. For each $i\in[t]$, set
\[
R_i:=\{r\in[R]:B_r\in\emb_i(\F_q^{n_i\times m_i})\},
\]
so that the $R_i$ partition $[R]$. Projecting onto block $i$,
\[
\cC_i=\pi_i(\cC)\subseteq\langle\pi_i(B_r):r\in[R]\rangle.
\]
For $r\notin R_i$ we have $\pi_i(B_r)=0$; for $r\in R_i$, $\pi_i(B_r)$ is a rank-one matrix in $\F_q^{n_i\times m_i}$. Hence $\cC_i\subseteq\langle\pi_i(B_r):r\in R_i\rangle$, giving $\trk(\cC_i)\le|R_i|$. Summing,
\[
\sum_{i=1}^{t}\trk(\cC_i)\le\sum_{i=1}^{t}|R_i|=R=\btrk(\cC).
\]

\emph{Upper bound: $\btrk(\cC)\le\sum_i\trk(\cC_i)$.}
For each $i\in[t]$, set $s_i:=\trk(\cC_i)$ and choose rank-one matrices $A_{i,1},\dots,A_{i,s_i}\in\F_q^{n_i\times m_i}$ whose $\F_q$-span contains $\cC_i$. The family
\[
\bigl\{\emb_i(A_{i,j}):i\in[t],\ j\in[s_i]\bigr\}\subseteq\cB(\bn,\bm)
\]
has cardinality $\sum_is_i$. For any $X=(X_1,\dots,X_t)\in\cC$, write $X_i=\sum_{j}\mu_{i,j}A_{i,j}$. Then $X=\sum_{i,j}\mu_{i,j}\emb_i(A_{i,j})$, showing $\btrk(\cC)\le\sum_is_i$.
\end{proof}

\subsection{Projection-wise and elementary bounds}

The identity gives an immediate projection-wise lower bound, complementary to the coordinate-code bound of Section~\ref{sec:coordinate-bounds}.

\begin{corollary}[Projection-wise lower bound]
\label{cor:projection-wise-bound}
For each $i\in[t]$ with $\cC_i:=\pi_i(\cC)\ne\{0\}$, write $k_i:=\dim_{\F_q}\cC_i$ and $d_i:=d_{\mathrm{rk}}(\cC_i)$ for the minimum rank distance of $\cC_i$. Then
\[
\btrk(\cC)\ge\sum_{i:\cC_i\ne 0}N_q(k_i,d_i)\ge\sum_{i:\cC_i\ne 0}(k_i+d_i-1).
\]
\end{corollary}

\begin{proof}
By Proposition~\ref{prop:blockwise-decomposition}, $\btrk(\cC)=\sum_i\trk(\cC_i)$. The coordinate-code argument of~\cite[Corollary~4.15]{BNRS19} gives $\trk(\cC_i)\ge N_q(k_i,d_i)\ge k_i+d_i-1$ for every nonzero $\cC_i$, and zero projections contribute zero to both sides.
\end{proof}

We also record an elementary upper bound, addressing how large $\btrk(\cC)$ can be in terms of the projection dimensions.

\begin{proposition}[Elementary upper bound]
\label{prop:elementary-upper-btrk}
Let $\cC\subseteq\cM(\bn,\bm)$ and set $k_i:=\dim\pi_i(\cC)$. Then
\[
\btrk(\cC)\le\sum_{i=1}^{t}\min\bigl\{n_im_i,\ k_i\min\{n_i,m_i\}\bigr\}.
\]
In particular, $\btrk(\cC)\le\sum_{i=1}^{t}n_im_i$, and if every block has size $2\times 2$, then $\btrk(\cC)\le 4t$.
\end{proposition}

\begin{proof}
By Proposition~\ref{prop:blockwise-decomposition}, it suffices to bound $\trk(\cC_i)$ for each block by the minimum of two quantities. First, choosing a basis $A_{i,1},\dots,A_{i,k_i}$ of $\cC_i$ and decomposing each $A_{i,j}$ into at most $\min(n_i,m_i)$ rank-one matrices gives $\trk(\cC_i)\le k_i\min\{n_i,m_i\}$. Second, the $n_im_i$ matrix units form a rank-one basis of $\F_q^{n_i\times m_i}\supseteq\cC_i$, giving $\trk(\cC_i)\le n_im_i$.
\end{proof}

\begin{remark}[Rank-metric specialization]
\label{rem:rank-spec}
When $t=1$, the ambient space is a single matrix space, block-simple tensors are exactly rank-one matrices, and Proposition~\ref{prop:blockwise-decomposition} reduces to $\btrk(\cC)=\trk(\cC)$ in the sense of~\cite[Definition~4.3]{BNRS19}.
\end{remark}

\begin{remark}[Hamming specialization]
\label{rem:hamming-spec}
When $n_i=m_i=1$ for all $i$, the ambient space is $\F_q^t$ with the Hamming metric, and a block-simple tensor is a scalar multiple of a coordinate vector $e_i$. The projection $\pi_i(\cC)\subseteq\F_q$ is either zero or all of $\F_q$, with $\trk(\F_q)=1$. Hence
\[
\btrk(\cC)=|\supp(\cC)|,\qquad\supp(\cC)=\{i\in[t]:\pi_i(\cC)\ne 0\}.
\]
For a nondegenerate Hamming code (full support), $\btrk(\cC)=t$.
\end{remark}

\subsection{Invariance of block tensor rank}

Although \(\btrk(\cC)\) is defined using block-simple tensors, it is an intrinsic invariant of the sum-rank metric space rather than an artefact of a
chosen coordinate representation. Indeed, bijective \(\F_q\)-linear sum-rank isometries preserve the elements of sum-rank weight one, and these are exactly the nonzero block-simple tensors. Hence they preserve block-simple spanning families and the minimum possible size of such families.

\begin{proposition}
\label{prop:btrk-invariance}
Let $\varphi\colon\cM(\bn,\bm)\to\cM(\bn',\bm')$ be a bijective $\F_q$-linear sum-rank isometry, and let $\cC':=\varphi(\cC)$. Then $\btrk(\cC)=\btrk(\cC')$.
\end{proposition}

\begin{proof}
A nonzero element of $\cM(\bn,\bm)$ has sum-rank weight one if and only if it is block-simple: if $B=\emb_i(vw^\top)$ with $v,w\ne 0$ then $\wt_{\SR}(B)=1$, while a tuple of sum-rank weight one must have a single nonzero block of rank one. Since $\varphi$ preserves weight, it maps block-simple tensors bijectively onto block-simple tensors. If $\cC$ lies in the span of $R$ block-simple tensors, so does $\varphi(\cC)$, and applying the same argument to $\varphi^{-1}$ gives the reverse inequality.
\end{proof}

\begin{example}
\label{ex:btrk}
Continuing Example~\ref{ex:gentuple}, consider
\[
\cC=\{(a_1I_2,\,a_2I_2):(a_1,a_2)\in\F_2^2\}
\subseteq \cM\bigl((2,2),(2,2)\bigr).
\]
The two block projections are
\(\pi_1(\cC)=\pi_2(\cC)=\langle I_2\rangle\subseteq\F_2^{2\times2}\). We claim that
\(\trk(\langle I_2\rangle)=2\). Indeed, \(\trk(\langle I_2\rangle)\ne1\),
because if \(\langle I_2\rangle\) were contained in the span of a single
rank-one matrix \(vw^\top\), then the nonzero matrix \(I_2\) would have rank at
most \(1\), a contradiction. Conversely,
\(I_2=e_1e_1^\top+e_2e_2^\top\), so
\(\langle I_2\rangle\subseteq\langle e_1e_1^\top,e_2e_2^\top\rangle\). Hence
\(\trk(\langle I_2\rangle)=2\). By Proposition~\ref{prop:blockwise-decomposition},
\[
\btrk(\cC)=\trk(\pi_1(\cC))+\trk(\pi_2(\cC))=2+2=4.
\]
On the other hand, \(\dim\cC=2\), and the nonzero codewords
\((I_2,0)\), \((0,I_2)\), and \((I_2,I_2)\) have sum-rank weights \(2\), \(2\),
and \(4\), respectively. Thus \(\dSR(\cC)=2\), so
\[
\dim\cC+\dSR(\cC)-1=2+2-1=3<4=\btrk(\cC).
\]
Therefore \(\cC\) is not BTR. The gap reflects the fact that each nonzero block
projection is generated by the rank-\(2\) matrix \(I_2\), so each projection
requires two rank-one matrices to span it; this block-internal rank cost is not
detected by the global Singleton-type lower bound.
\end{example}

\section{Coordinate-Code Bounds and Extremal Notions}
\label{sec:coordinate-bounds}

We now derive lower bounds for \(\btrk(\cC)\) from the global parameters
\((\dim\cC,\dSR(\cC))\). The lower-bound method is the sum-rank analogue of the
coordinate-code argument of~\cite{BNRS19}: a block-simple basis for \(\cC\)
induces a linear block code over \(\F_q\) whose Hamming weight controls the
sum-rank weight, and classical Hamming-metric bounds transfer back to
\(\btrk(\cC)\). This global method is complementary to the projection-wise bound
of Corollary~\ref{cor:projection-wise-bound}. The resulting extremal criteria
will be used in Section~\ref{sec:constructions} to construct BTR and
block-tensor-rank-extremal codes.

\subsection{Block-simple bases and coordinate codes}
\label{subsec:coordinate-codes}

We first associate to a block-simple spanning family a usual Hamming-metric
coordinate code. This is the bridge between block tensor rank and classical
bounds for linear block codes.

\begin{definition}
\label{def:simple-basis}
Let \(\cC\) be a nonzero \([\bn\times\bm,k]_q\) sum-rank code, and set
\(R:=\btrk(\cC)\). A \emph{block-simple basis} for \(\cC\) is a set
\(\cB=\{B_1,\dots,B_R\}\subseteq\cB(\bn,\bm)\) of \(\F_q\)-linearly independent
block-simple tensors such that \(\cC\subseteq\langle B_1,\dots,B_R\rangle\).
For each \(r\in[R]\), let \(i_r\in[t]\) be the unique block on which \(B_r\) is
supported, and for each \(i\in[t]\), set \(R_i:=\{r\in[R]:i_r=i\}\).
\end{definition}

A block-simple basis exists. Indeed, by the definition of \(R=\btrk(\cC)\),
there are \(R\) block-simple tensors whose span contains \(\cC\). Such a family
must be linearly independent: if one tensor were redundant, then removing it
would give a block-simple spanning family of size smaller than \(R\), contrary
to the minimality of \(R\).

Fix a block-simple basis \(\cB=\{B_1,\dots,B_R\}\). Since the elements of
\(\cB\) are linearly independent, every element of \(\langle\cB\rangle\) has a
unique expression \(\sum_{r=1}^{R}\lambda_rB_r\). Thus, we have a linear
isomorphism
\[
\psi_\cB:\langle\cB\rangle\xrightarrow{\sim}\F_q^R,\qquad
\sum_{r=1}^{R}\lambda_rB_r\longmapsto(\lambda_1,\dots,\lambda_R).
\]
Restricting this map to \(\cC\subseteq\langle\cB\rangle\) gives an injective
\(\F_q\)-linear map \(\psi_\cB|_\cC:\cC\hookrightarrow\F_q^R\).

\begin{definition}
\label{def:coord-block-code}
The \emph{coordinate block code} of \(\cC\) with respect to the block-simple
basis \(\cB\) is \(\cC_\cB:=\psi_\cB(\cC)\subseteq\F_q^R\). Equivalently,
\(\cC_\cB\) records the coordinates of codewords of \(\cC\) in the
block-simple basis \(\cB\).
\end{definition}

The following lemma compares the sum-rank weight of a codeword with the
Hamming weight of its coordinate vector in \(\cC_\cB\).

\begin{lemma}
\label{lem:weight-comparison}
With notation as above, write each block-simple tensor as
\(B_r=\emb_{i_r}(v_rw_r^\top)\), where \(v_r\in\F_q^{n_{i_r}}\) and
\(w_r\in\F_q^{m_{i_r}}\). If \(M\in\cC\) has expansion
\(M=\sum_{r=1}^{R}\lambda_rB_r\), then
\[
\wt_{\SR}(M)\le \wt_H(\psi_\cB(M)).
\]
\end{lemma}

\begin{proof}
For each block \(i\in[t]\), the \(i\)-th component of \(M\) is
\[
M_i=\sum_{r\in R_i}\lambda_r v_rw_r^\top.
\]
This is a sum of at most \(\#\{r\in R_i:\lambda_r\ne0\}\) rank-one matrices.
By subadditivity of rank,
\[
\rk(M_i)\le \#\{r\in R_i:\lambda_r\ne0\}.
\]
Since the sets \(R_i\) partition \([R]\), summing over all blocks gives
\[
\wt_{\SR}(M)
=
\sum_{i=1}^{t}\rk(M_i)
\le
\sum_{i=1}^{t}\#\{r\in R_i:\lambda_r\ne0\}
=
\#\{r\in[R]:\lambda_r\ne0\}.
\]
But \(\psi_\cB(M)=(\lambda_1,\dots,\lambda_R)\), so the last quantity is
\(\wt_H(\psi_\cB(M))\). This proves the lemma.
\end{proof}

Recall that \(N_q(k,d)\) denotes the minimum length of a linear
\([N,k,\ge d]_q\) block code. By the classical Singleton bound,
\(N_q(k,d)\ge k+d-1\), with equality precisely when an MDS
\([k+d-1,k,d]_q\) block code exists over \(\F_q\).

\subsection{Coordinate-code bounds and extremal codes}
\label{subsec:global-bounds-extremal}

We now apply the coordinate block code to obtain global lower bounds on \(\btrk(\cC)\). The fundamental estimate is the coordinate-code bound,
which compares \(\btrk(\cC)\) with the minimum possible length of a classical linear block code having the same dimension and minimum distance.
The Singleton and Griesmer bounds are then obtained as immediate consequences of the corresponding classical bounds in the Hamming metric. These bounds lead to two extremal classes: one at the \(N_q(k,d)\)-level and one at the Singleton level \(k+d-1\).

\begin{theorem}[Coordinate-code bound]
\label{thm:main-bound}
Let \(\cC\) be a nonzero \([\bn\times\bm,k,d]_q\) sum-rank code. Then
\[
    \btrk(\cC)\ge N_q(k,d).
\]
\end{theorem}

\begin{proof}
Set \(R:=\btrk(\cC)\), and let \(\cB\) be a block-simple basis of size \(R\). By Lemma~\ref{lem:weight-comparison}, every nonzero \(M\in\cC\) satisfies
\[
    d\le \wt_{\SR}(M)\le \wt_H(\psi_\cB(M)).
\]
Thus, \(\cC_\cB\subseteq\F_q^R\) is a linear \([R,k,\ge d]_q\) block code. By the definition of \(N_q(k,d)\), we get
\[
    R\ge N_q(k,d).
\]
Since \(R=\btrk(\cC)\), the theorem follows.
\end{proof}

\begin{corollary}[Singleton coordinate-code bound]
\label{cor:basic-btrk-bound}
For every nonzero sum-rank code \(\cC\),
\[
    \btrk(\cC)\ge \dim(\cC)+\dSR(\cC)-1.
\]
Equivalently, if \(\cC\) is a nonzero \([\bn\times\bm,k,d]_q\) sum-rank code,
then
\[
    \btrk(\cC)\ge k+d-1.
\]
\end{corollary}

\begin{proof}
By Theorem~\ref{thm:main-bound},
\[
    \btrk(\cC)\ge N_q(k,d).
\]
The classical Singleton bound for linear block codes gives
\[
    N_q(k,d)\ge k+d-1.
\]
Therefore
\[
    \btrk(\cC)\ge k+d-1.
\]
Taking \(k=\dim(\cC)\) and \(d=\dSR(\cC)\) gives the stated form.
\end{proof}

We next refine the Singleton coordinate-code bound using the Griesmer bound for linear Hamming-metric codes~\cite{Griesmer1960}. For positive integers
\(k\) and \(d\), define
\[
    \Gri(k,d):=
    \sum_{j=0}^{k-1}
    \left\lceil\frac{d}{q^j}\right\rceil.
\]

\begin{corollary}[Griesmer coordinate-code bound]
\label{cor:griesmer-btrk-bound}
For every nonzero \([\bn\times\bm,k,d]_q\) sum-rank code \(\cC\),
\[
    \btrk(\cC)\ge \Gri(k,d)\ge k+d-1.
\]
\end{corollary}

\begin{proof}
Theorem~\ref{thm:main-bound} gives
\[
    \btrk(\cC)\ge N_q(k,d).
\]
By the Griesmer bound for linear block codes,
\[
    N_q(k,d)\ge \Gri(k,d),
\]
and hence
\[
    \btrk(\cC)\ge \Gri(k,d).
\]
Finally, the \(j=0\) term in \(\Gri(k,d)\) is \(d\), while each of the remaining \(k-1\) terms is at least \(1\). Therefore
\[
    \Gri(k,d)\ge d+(k-1)=k+d-1.
\]
\end{proof}

The Griesmer refinement is strict whenever \(k\ge 2\) and \(d>q\), since then
\[
    \left\lceil \frac{d}{q}\right\rceil\ge 2.
\]
For example, if \(q=2\), \(k=3\), and \(d=4\), then
\[
    k+d-1=6,
    \qquad
    \Gri(3,4)=4+2+1=7.
\]

The global lower bound and the projection-wise lower bound capture different aspects of the code. Theorem~\ref{thm:main-bound} depends only on the global
parameters \(k\) and \(\dSR(\cC)\), while Corollary~\ref{cor:projection-wise-bound} detects rank-metric structure inside
the individual projections. Thus, one may combine them as
\[
    \btrk(\cC)\ge
    \max\Biggl\{
        N_q(k,\dSR(\cC)),
        \sum_{i:\pi_i(\cC)\ne0}
        N_q\bigl(
            \dim\pi_i(\cC),
            d_{\rm rk}(\pi_i(\cC))
        \bigr)
    \Biggr\}.
\]

Combining Corollary~\ref{cor:basic-btrk-bound} with the elementary upper bound of Proposition~\ref{prop:elementary-upper-btrk} also gives a useful distance estimate. If \(k_i:=\dim\pi_i(\cC)\), then
\[
    d
    \le
    \btrk(\cC)-k+1
    \le
    \sum_{i=1}^{t}
    \min\{n_im_i,\ k_i\min(n_i,m_i)\}
    -k+1.
\]
In particular, for \(q=2\) and all blocks of size \(2\times2\), this yields
\[
    d\le 4t-k+1.
\]

\begin{definition}
\label{def:btr}
A nonzero \([\bn\times\bm,k,d]_q\) sum-rank code \(\cC\) is called
\emph{block-tensor-rank-extremal} if
\[
    \btrk(\cC)=N_q(k,d).
\]
It is called \emph{block tensor rank minimum}, or \emph{BTR}, if
\[
    \btrk(\cC)=k+d-1.
\]
\end{definition}

Since \(N_q(k,d)\ge k+d-1\), every BTR code is block-tensor-rank-extremal
when
\[
    N_q(k,d)=k+d-1,
\]
equivalently when an MDS \([k+d-1,k,d]_q\) block code exists over \(\F_q\). In general, block-tensor-rank-extremality is the more flexible condition,
while BTR is the stronger Singleton-level optimality condition.

\begin{proposition}
\label{prop:btr-mds}
If \(\cC\) is a BTR \([\bn\times\bm,k,d]_q\) code and \(\cB\) is a
block-simple basis for \(\cC\), then \(\cC_\cB\) is an MDS
\([k+d-1,k,d]_q\) block code.
\end{proposition}

\begin{proof}
Let \(R:=\btrk(\cC)\). Since \(\cC\) is BTR, we have
\[
    R=k+d-1.
\]
By Lemma~\ref{lem:weight-comparison}, the coordinate block code \(\cC_\cB\subseteq\F_q^R\) is a linear \([R,k,\ge d]_q\) block code.
The Singleton bound gives
\[
    d_H(\cC_\cB)
    \le
    R-k+1
    =
    (k+d-1)-k+1
    =
    d.
\]
Since \(d_H(\cC_\cB)\ge d\), we have
\[
    d_H(\cC_\cB)=d.
\]
Thus, \(\cC_\cB\) is an MDS \([k+d-1,k,d]_q\) block code.
\end{proof}

\begin{corollary}
\label{cor:btr-implies-mds}
If a BTR \([\bn\times\bm,k,d]_q\) sum-rank code exists, then an MDS \([k+d-1,k,d]_q\) linear block code over \(\F_q\) exists.
\end{corollary}

\begin{proof}
This follows immediately from Proposition~\ref{prop:btr-mds}.
\end{proof}

\section{Constructions of Block-Tensor-Rank-Extremal Codes}
\label{sec:constructions}

\subsection{The block-diagonal lifting map}
\label{subsec:block-diagonal-lifting}

We now introduce the lifting mechanism used in all subsequent constructions.
Starting from a classical linear block code over \(\F_q\), we divide its
coordinates into groups, one for each sum-rank block. Each coordinate group is
then converted into a matrix block by a bilinear diagonal construction. The
precise map and its basic properties are given below.

\begin{definition}
\label{def:phi}
Let \(\br=(r_1,\dots,r_t)\) be a composition of \(R\ge1\), so \(r_i\ge0\) and
\(\sum_i r_i=R\). For each \(i\in[t]\), let
\(V_i\in\F_q^{n_i\times r_i}\) and \(W_i\in\F_q^{m_i\times r_i}\), and write
\(\bV=(V_1,\dots,V_t)\), \(\bW=(W_1,\dots,W_t)\). For
\(c=(c^{(1)},\dots,c^{(t)})\in\F_q^R\), with \(c^{(i)}\in\F_q^{r_i}\), define
the \emph{block-diagonal lifting map}
\[
\phi^\br_{\bV,\bW}:\F_q^R\to\cM(\bn,\bm),\qquad
c\longmapsto (V_i\diag(c^{(i)})W_i^\top)_{i=1}^{t}.
\]
If \(r_i=0\), the \(i\)-th block is understood to be the zero matrix.
\end{definition}

In this construction, we call the block code \(C\subseteq\F_q^R\) the \emph{backbone code}, and the image
\(\cC=\phi^\br_{\bV,\bW}(C)\) the \emph{lifted sum-rank code}.

\begin{lemma}
\label{lem:phi-props}
The map \(\phi^\br_{\bV,\bW}\) is \(\F_q\)-linear, and for every \(c\in\F_q^R\), $\wt_{\SR}(\phi^\br_{\bV,\bW}(c))\le \wt_H(c).$
\end{lemma}

\begin{proof}
Linearity is immediate. Write the columns of \(V_i\) and \(W_i\) as
\(v_{i,1},\dots,v_{i,r_i}\) and \(w_{i,1},\dots,w_{i,r_i}\). Then
\[
V_i\diag(c^{(i)})W_i^\top
=
\sum_{r:c^{(i)}_r\ne0}c^{(i)}_r\,v_{i,r}w_{i,r}^\top,
\]
which is a sum of \(\wt_H(c^{(i)})\) rank-one matrices. Hence
\[
\rk(V_i\diag(c^{(i)})W_i^\top)\le \wt_H(c^{(i)}).
\]
Summing over \(i\) gives
\[
\wt_{\SR}(\phi^\br_{\bV,\bW}(c))
\le
\sum_{i=1}^{t}\wt_H(c^{(i)})
=
\wt_H(c).
\]
\end{proof}

\begin{proposition}[General lifting]
\label{prop:phi-construction}
Let \(C\subseteq\F_q^R\) be a linear \([R,k,d]_q\) block code and let
\((\br,\bV,\bW)\) be lifting data as above. Suppose
\(\phi^\br_{\bV,\bW}|_C\) is injective, and set
\(\cC:=\phi^\br_{\bV,\bW}(C)\subseteq\cM(\bn,\bm)\). Then
\(\dim\cC=k\) and \(\btrk(\cC)\le R\). If every nonzero \(c\in C\) satisfies
\begin{equation}
\label{eq:distance-lift}
\sum_{i=1}^{t}
\rk\bigl(V_i\diag(c^{(i)})W_i^\top\bigr)\ge d,
\end{equation}
then \(\dSR(\cC)\ge d\). Furthermore, if \(R=k+d-1\) and
\eqref{eq:distance-lift} holds, then \(\cC\) is a BTR
\([\bn\times\bm,k,d]_q\) code.
\end{proposition}

\begin{proof}
For each \(i\in[t]\) and \(r\in[r_i]\), the tensor
\(\emb_i(v_{i,r}w_{i,r}^\top)\) is block-simple. There are
\(\sum_i r_i=R\) such tensors, and for any \(c\in C\),
\[
\phi^\br_{\bV,\bW}(c)
=
\sum_{i=1}^{t}\emb_i\bigl(V_i\diag(c^{(i)})W_i^\top\bigr)
=
\sum_{i=1}^{t}\sum_{r=1}^{r_i}
c^{(i)}_r\,\emb_i(v_{i,r}w_{i,r}^\top).
\]
Thus, \(\cC\) is contained in the span of \(R\) block-simple tensors, so
\(\btrk(\cC)\le R\). Since \(\phi^\br_{\bV,\bW}|_C\) is injective and
\(\dim C=k\), we have \(\dim\cC=k\). If \eqref{eq:distance-lift} holds for
every nonzero \(c\in C\), then every nonzero lifted codeword has sum-rank
weight at least \(d\), and hence \(\dSR(\cC)\ge d\).

Assume now that \(R=k+d-1\). By Corollary~\ref{cor:basic-btrk-bound},
\[
\btrk(\cC)\ge k+d-1=R.
\]
Together with \(\btrk(\cC)\le R\), this gives \(\btrk(\cC)=R\). Since \(C\) has
minimum Hamming distance \(d\), there exists \(c\in C\) with \(\wt_H(c)=d\).
By Lemma~\ref{lem:phi-props},
\[
\wt_{\SR}(\phi^\br_{\bV,\bW}(c))\le d.
\]
Together with \(\dSR(\cC)\ge d\), this gives \(\dSR(\cC)=d\). Therefore
\(\cC\) is BTR.
\end{proof}

\begin{definition}
\label{def:extremal-tuple}
An \emph{extremal tuple} is a quadruple \((C,\br,\bV,\bW)\), where
\(C\subseteq\F_q^R\) is an MDS \([R,k,d]_q\) block code and the lifting data
\((\br,\bV,\bW)\) satisfy the hypotheses of
Proposition~\ref{prop:phi-construction} with \(R=k+d-1\). The induced
sum-rank code
\[
\cC:=\phi^\br_{\bV,\bW}(C)\subseteq\cM(\bn,\bm)
\]
is then BTR. We call \(\cC\) an \emph{extremal BTR code}.
\end{definition}

\subsection{Three explicit regimes}
\label{subsec:three-regimes}

We now describe three parameter regimes in which extremal tuples exist. Throughout this
subsection, \(R:=k+d-1\). The first regime is the simplest one: all \(R\)
coordinates of the MDS backbone are placed inside a single sufficiently large matrix block.

\begin{proposition}[Fat-block regime]
\label{prop:fat-block}
Let \(R=k+d-1\), and suppose there exists \(i_0\in[t]\) such that
\(\min(n_{i_0},m_{i_0})\ge R\). Then, for every MDS \([R,k,d]_q\) block code
\(C\), an extremal tuple \((C,\br,\bV,\bW)\) exists. Consequently, if an MDS
\([R,k,d]_q\) block code over \(\F_q\) exists, then a BTR
\([\bn\times\bm,k,d]_q\) sum-rank code exists.
\end{proposition}

\begin{proof}
Set \(r_{i_0}=R\) and \(r_i=0\) for \(i\ne i_0\). Since
\(\min(n_{i_0},m_{i_0})\ge R\), choose full-column-rank matrices
\(V_{i_0}\in\F_q^{n_{i_0}\times R}\) and
\(W_{i_0}\in\F_q^{m_{i_0}\times R}\). For all other blocks, the lifting data
are empty. For \(c\in C\setminus\{0\}\), left multiplication by \(V_{i_0}\) and
right multiplication by \(W_{i_0}^\top\) preserve the rank of \(\diag(c)\), so
\[
\rk\bigl(V_{i_0}\diag(c)W_{i_0}^{\top}\bigr)
=
\rk(\diag(c))
=
\wt_H(c)
\ge d.
\]
Thus, \eqref{eq:distance-lift} holds. The same equality shows that
\(\phi^\br_{\bV,\bW}(c)=0\) implies \(c=0\), so the restricted lifting map is
injective. Proposition~\ref{prop:phi-construction} gives the claim.
\end{proof}

The next regime distributes the \(R\) coordinates of the MDS backbone among
several blocks. To ensure that enough Hamming weight is converted into
sum-rank weight, we use matrices whose columns satisfy an MDS-type independence
condition.

\begin{definition}
\label{def:n-mds}
A matrix \(V\in\F_q^{n\times r}\) is \emph{\(n\)-MDS} if every
\(\min(n,r)\) columns of \(V\) are linearly independent. Equivalently, if
\(n\le r\), then \(V\) is a generator matrix of an MDS \([r,n,r-n+1]_q\) code;
if \(n\ge r\), then \(V\) has full column rank.
\end{definition}

\begin{proposition}[MDS-backbone regime]
\label{prop:mds-backbone}
Let \(R=k+d-1\), let \(C\subseteq\F_q^R\) be an MDS \([R,k,d]_q\) block code,
and let \(\br=(r_1,\dots,r_t)\) be a composition of \(R\). For each \(i\), let
\(V_i\in\F_q^{n_i\times r_i}\) be \(n_i\)-MDS and
\(W_i\in\F_q^{m_i\times r_i}\) be \(m_i\)-MDS. Suppose every nonzero
\(c\in C\), with \(w_i:=\wt_H(c^{(i)})\), satisfies
\begin{equation}
\label{eq:mds-backbone-condition}
\sum_{i=1}^{t}
\max\bigl\{0,\ \min(n_i,w_i)+\min(m_i,w_i)-w_i\bigr\}
\ge d.
\end{equation}
Then \((C,\br,\bV,\bW)\) is an extremal tuple.
\end{proposition}

\begin{proof}
Fix \(c\in C\setminus\{0\}\). For each \(i\), let \(S_i\subseteq[r_i]\) be the
support of \(c^{(i)}\), so \(|S_i|=w_i\). The matrix \(V_i\diag(c^{(i)})\) has
the same column space as the submatrix \((V_i)_{S_i}\). Since \(V_i\) is
\(n_i\)-MDS,
\[
\rk\bigl(V_i\diag(c^{(i)})\bigr)=\min(n_i,w_i).
\]
Similarly,
\[
\rk\bigl(\diag(c^{(i)})W_i^\top\bigr)=\min(m_i,w_i).
\]
Applying Frobenius's rank inequality to
\(A=V_i\), \(B=\diag(c^{(i)})\), and \(C=W_i^\top\), and using
\(\rk(\diag(c^{(i)}))=w_i\), gives
\[
\rk\bigl(V_i\diag(c^{(i)})W_i^\top\bigr)
\ge
\min(n_i,w_i)+\min(m_i,w_i)-w_i.
\]
Since rank is nonnegative,
\[
\rk\bigl(V_i\diag(c^{(i)})W_i^\top\bigr)
\ge
\max\bigl\{0,\min(n_i,w_i)+\min(m_i,w_i)-w_i\bigr\}.
\]
Summing over \(i\) and using \eqref{eq:mds-backbone-condition}, we obtain
\eqref{eq:distance-lift}. In particular, the lifted image of every nonzero
\(c\in C\) is nonzero, so \(\phi^\br_{\bV,\bW}|_C\) is injective.
Proposition~\ref{prop:phi-construction} completes the proof.
\end{proof}

A useful uniform consequence is obtained when all blocks have the same row and
column sizes.

\begin{corollary}[Uniform case with \(n\ge d\)]
\label{cor:uniform-ngeqd}
Assume \(n_i=n\) and \(m_i=m\) for all \(i\), with \(d\le n\le m\). Let \(C\)
be an MDS \([R,k,d]_q\) block code with \(R=k+d-1\), and let
\(\br=(r,\dots,r)\) with \(tr=R\). If \(n+m\ge r+d\), each \(V_i\) is
\(n\)-MDS, and each \(W_i\) is \(m\)-MDS, then \((C,\br,\bV,\bW)\) is an
extremal tuple.
\end{corollary}

\begin{proof}
We verify \eqref{eq:mds-backbone-condition}. Let \(c\in C\setminus\{0\}\),
set \(w_i:=\wt_H(c^{(i)})\), and let \(w:=\sum_iw_i=\wt_H(c)\). Since \(C\)
has minimum distance \(d\), \(w\ge d\). Define
\[
s_i:=\max\{0,\min(n,w_i)+\min(m,w_i)-w_i\}.
\]
If \(w_i\le n\), then \(s_i=w_i\). If \(n<w_i\le m\), then \(s_i=n\). If
\(w_i>m\), then \(s_i=\max\{0,n+m-w_i\}\), and since \(w_i\le r\) and
\(n+m\ge r+d\), we get \(s_i\ge n+m-r\ge d\). Thus, either some block already
contributes at least \(d\), or all blocks have \(w_i\le n\), in which case
\(\sum_i s_i=\sum_iw_i=w\ge d\). Therefore
\eqref{eq:mds-backbone-condition} holds, and Proposition~\ref{prop:mds-backbone}
applies.
\end{proof}

The preceding uniform corollary does not cover the genuinely sum-rank
\emph{pooling regime}, where the target distance is larger than the maximum
rank contribution available from a single block. More precisely, this is the
range \(\max_i n_i<d\le\sum_i n_i\). The following result treats a tractable
subregime in which each block receives only a small number of backbone
coordinates.

\begin{corollary}[Uniform pooling regime with \(r\le n\)]
\label{cor:uniform-pooling}
Assume \(n_i=n\) and \(m_i=m\) for all \(i\), with \(n<d\le tn\). Let
\(\br=(r,\dots,r)\), where \(tr=R=k+d-1\). Suppose \(r\le\min(n,m)\), and each
\(V_i\in\F_q^{n\times r}\) and \(W_i\in\F_q^{m\times r}\) has full column rank
\(r\). Then for every MDS \([R,k,d]_q\) block code \(C\), the tuple
\((C,\br,\bV,\bW)\) is extremal. In particular, whenever an MDS \([R,k,d]_q\)
block code over \(\F_q\) exists, BTR \([\bn\times\bm,k,d]_q\) codes exist in
this regime.
\end{corollary}

\begin{proof}
Since \(r\le n\) and \(r\le m\), full column rank of \(V_i\) and \(W_i\) is
equivalent to the \(n\)-MDS and \(m\)-MDS conditions needed in
Proposition~\ref{prop:mds-backbone}. Let \(c\in C\setminus\{0\}\), and set
\(w_i:=\wt_H(c^{(i)})\). Because \(w_i\le r\le\min(n,m)\),
\[
\min(n,w_i)=\min(m,w_i)=w_i,
\]
and hence the \(i\)-th summand in \eqref{eq:mds-backbone-condition} is \(w_i\).
Summing over \(i\) gives \(\sum_iw_i=\wt_H(c)\ge d\). Therefore
\eqref{eq:mds-backbone-condition} holds, and Proposition~\ref{prop:mds-backbone}
gives the result.
\end{proof}

The hypotheses of Corollary~\ref{cor:uniform-pooling} imply a small-dimension
condition. Indeed, \(n<d\), \(r\le n\), and \(tr=R=k+d-1\) imply
\[
k+d-1\le tn,\qquad\text{or equivalently}\qquad k\le tn-d+1.
\]
Compared with the MSRD Singleton-type bound \(k\le m(tn-d+1)\) from
Theorem~\ref{thm:singleton}, this restricts the dimension by a factor of \(m\).
Thus, Corollary~\ref{cor:uniform-pooling} covers BTR codes of relatively small
dimension in the pooling regime; treating larger \(k\) requires additional
support-distribution control for the MDS backbone.

\begin{example}
\label{ex:uniform-pooling}
Take \(q=5\), \(t=2\), \(n=m=2\), \(r=2\), \(k=2\), and \(d=3\). Then
\(R=tr=4=k+d-1\). Let \(C\subseteq\F_5^4\) be the Reed--Solomon \([4,2,3]_5\)
code with evaluation points \((1,2,3,4)\), and take
\(V_1=V_2=W_1=W_2=I_2\). The lifting is
\[
\mathcal{C}=\phi^\br_{\bV,\bW}(c_1,c_2,c_3,c_4)
=
(\diag(c_1,c_2),\,\diag(c_3,c_4)).
\]
For example, \(c=(0,1,2,3)\in C\) maps to
\[
(\diag(0,1),\,\diag(2,3)),
\]
which has sum-rank weight \(1+2=3=d\). Hence
\(\btrk(\mathcal{C})=k+d-1=4\).
\end{example}


\subsection{Vandermonde-lifted Reed--Solomon BTR codes}
\label{sec:btr-family}

We now provide an explicit parameterised BTR family fitting the uniform
MDS-backbone regime of Corollary~\ref{cor:uniform-ngeqd}. The construction
starts from a Reed--Solomon MDS backbone and uses Vandermonde matrices as the
lifting matrices.

\begin{definition}
\label{def:vrs}
Let \(q\) be a prime power and fix positive integers \(t,n,m,k,d,r\) with
\(R:=k+d-1=tr\), \(d\le n\le m\), \(n+m\ge r+d\), and
\(q\ge\max(R,r+1)\). Let \(C\subseteq\F_q^R\) be a Reed--Solomon
\([R,k,d]_q\) code with \(R\) distinct evaluation points. For each \(i\in[t]\),
let \(V_i\in\F_q^{n\times r}\) and \(W_i\in\F_q^{m\times r}\) be Vandermonde
matrices
\[
V_i=[u_{i,s}^{\rho-1}]_{\rho\in[n],s\in[r]},
\qquad
W_i=[v_{i,s}^{\rho-1}]_{\rho\in[m],s\in[r]},
\]
where \(u_{i,1},\dots,u_{i,r}\in\F_q^*\) and
\(v_{i,1},\dots,v_{i,r}\in\F_q^*\) are pairwise distinct. The associated
\emph{Vandermonde-lifted Reed--Solomon sum-rank code} is
\(\cC:=\phi^\br_{\bV,\bW}(C)\subseteq\cM(\bn,\bm)\), with
\(\br=(r,\dots,r)\), \(\bn=(n,\dots,n)\), and \(\bm=(m,\dots,m)\).
\end{definition}

\begin{theorem}
\label{thm:vrs-btr}
Every Vandermonde-lifted Reed--Solomon sum-rank code is a BTR
\([\bn\times\bm,k,d]_q\) sum-rank code.
\end{theorem}

\begin{proof}
We verify the hypotheses of Corollary~\ref{cor:uniform-ngeqd}. Fix \(i\), and
let \(S\subseteq[r]\) have size \(p:=\min(n,r)\). The first \(p\) rows of
\((V_i)_S\) form the square Vandermonde matrix
\[
[u_{i,s}^{\rho-1}]_{\rho\in[p],s\in S},
\]
whose determinant is
\[
\prod_{\substack{s<s'\\ s,s'\in S}}(u_{i,s'}-u_{i,s})\ne0.
\]
Thus, every \(p\) columns of \(V_i\) are linearly independent, so \(V_i\) is
\(n\)-MDS. The same argument shows that \(W_i\) is \(m\)-MDS. Since \(C\) is a
Reed--Solomon \([R,k,d]_q\) code, it is MDS, and the assumption \(q\ge R\)
ensures that \(R\) distinct evaluation points are available. The remaining
hypotheses \(d\le n\le m\) and \(n+m\ge r+d\) are part of
Definition~\ref{def:vrs}. Therefore \((C,\br,\bV,\bW)\) is an extremal tuple,
and \(\cC=\phi^\br_{\bV,\bW}(C)\) is BTR.
\end{proof}

\begin{example}
\label{ex:vrs}
Take \(q=7\), \(t=2\), \(n=m=3\), \(k=2\), \(d=3\), \(r=2\), and \(R=4\). Let
\(C\) be the Reed--Solomon \([4,2,3]_7\) code generated by
\[
G_C=
\begin{pmatrix}
1&1&1&1\\
1&2&3&4
\end{pmatrix}.
\]
Choose \(V_1=V_2\) with nodes \((1,2)\), and \(W_1=W_2\) with nodes
\((1,3)\):
\[
V_i=
\begin{pmatrix}
1&1\\
1&2\\
1&4
\end{pmatrix},
\qquad
W_i=
\begin{pmatrix}
1&1\\
1&3\\
1&2
\end{pmatrix}.
\]
Each is \(3\)-MDS. Hence the resulting code
\(\cC\subseteq\F_7^{3\times3}\oplus\F_7^{3\times3}\) is BTR with parameters
\([(3,3)\times(3,3),2,3]_7\) and \(\btrk(\cC)=4\).
\end{example}

We finally relate this construction to rank-metric and linearized
Reed--Solomon constructions. When \(t=1\), Theorem~\ref{thm:vrs-btr} produces a
rank-metric code \(\cC\subseteq\F_q^{n\times m}\) generated by Vandermonde
lifts of a Reed--Solomon code, with parameters \([n\times m,k,d]_q\) and
\(\btrk(\cC)=k+d-1\). This is a base-field analogue of the MTR phenomenon
studied in~\cite[Theorem~5.7]{BNRS19}; Definition~\ref{def:vrs} works over
\(\F_q\) throughout, using classical rather than skew Vandermonde structure.

We also recall how this viewpoint compares with linearized Reed--Solomon
codes. Fix an \(\F_q\)-basis
\(\boldsymbol\beta=(\beta_1,\ldots,\beta_m)\) of \(\F_{q^m}\), and let
\[
\mathrm{ext}_{\boldsymbol\beta}:\F_{q^m}^n\longrightarrow \F_q^{m\times n}
\]
be the componentwise expansion map sending each coordinate
\(x\in\F_{q^m}\) to its coordinate column with respect to
\(\boldsymbol\beta\). More generally, for sum-rank block structures, this
expansion is applied blockwise.

Let \(C_{\rm LRS}\subseteq\F_{q^m}^n\) be an \(\F_{q^m}\)-linear linearized
Reed--Solomon code of dimension \(K\), obtained from skew-polynomial
evaluations at \(\ell=t\) distinct \(\sigma\)-conjugacy classes in the standard
notation of~\cite{MartinezPenas2018}. Its componentwise expansion
\[
\mathcal C_{\rm LRS}:=\mathrm{ext}_{\boldsymbol\beta}(C_{\rm LRS})
\]
is an \(\F_q\)-linear sum-rank code of dimension \(Km\) and minimum sum-rank
distance \(D=n-K+1\). If \(\mathcal C_{\rm LRS}\) were BTR, then
Corollary~\ref{cor:btr-implies-mds} would imply the existence of an MDS block code
\[
[Km+n-K,\;Km,\;n-K+1]_q
\]
over \(\F_q\). Thus, over small base fields, the obstruction
\[
\btrk(\mathcal C_{\rm LRS})\ge N_q(Km,n-K+1)
\]
may be strictly stronger than the Singleton value \(Km+n-K\). This shows that
MSRD optimality and BTR optimality are distinct notions.

\subsection{Griesmer-optimal backbones and extremal liftings}
\label{subsec:griesmer-backbones}

The Vandermonde-lifted Reed--Solomon construction is BTR because it starts
from an MDS backbone of length \(k+d-1\). Over small base fields such a
backbone may not exist, and the Griesmer bound can give the stronger
obstruction \(N_q(k,d)\ge\Gri(k,d)>k+d-1\). The next result shows that, when a
Griesmer-optimal Hamming code is available, the same block-diagonal lifting
attains the sharp value \(\btrk=N_q(k,d)\).

\begin{proposition}[Griesmer-optimal backbone liftings]
\label{prop:griesmer-backbone-lifting}
Let \(C\subseteq\F_q^R\) be a linear \([R,k,d]_q\) block code with
\(R=\Gri(k,d)\), and let \(\br=(r_1,\dots,r_t)\) be a composition of \(R\) with
\(r_i\le\min(n_i,m_i)\) for every \(i\). For each \(i\), let
\(V_i\in\F_q^{n_i\times r_i}\) and \(W_i\in\F_q^{m_i\times r_i}\) have full
column rank. Then \(\cC:=\phi^\br_{\bV,\bW}(C)\subseteq\cM(\bn,\bm)\) satisfies
\[
\dim_{\F_q}\cC=k,\qquad
\dSR(\cC)=d,\qquad
\btrk(\cC)=\Gri(k,d)=N_q(k,d).
\]
Hence \(\cC\) is block-tensor-rank-extremal. It is BTR precisely when
\(\Gri(k,d)=k+d-1\).
\end{proposition}

\begin{proof}
For \(c\in C\), let \(S_i\subseteq[r_i]\) be the support of \(c^{(i)}\), with
\(|S_i|=w_i\). Since \(r_i\le\min(n_i,m_i)\) and \(V_i,W_i\) have full column
rank, the submatrices indexed by \(S_i\) have full column rank. Therefore
\[
\rk(V_i\diag(c^{(i)})W_i^\top)=w_i=\wt_H(c^{(i)}).
\]
Summing over \(i\),
\[
\wt_{\SR}(\phi^\br_{\bV,\bW}(c))
=
\sum_{i=1}^{t}\wt_H(c^{(i)})
=
\wt_H(c).
\]
Thus, the lifting is injective on \(C\), and it preserves the minimum distance:
\(\dSR(\cC)=d\). The lifted code lies in the span of \(R\) block-simple tensors,
so \(\btrk(\cC)\le R\). Corollary~\ref{cor:griesmer-btrk-bound} gives
\(\btrk(\cC)\ge\Gri(k,d)=R\), hence \(\btrk(\cC)=R\). Finally, the Griesmer
bound gives \(N_q(k,d)\ge\Gri(k,d)=R\), while the existence of the length-\(R\)
code \(C\) gives \(N_q(k,d)\le R\). Therefore \(N_q(k,d)=R\), and the result
follows.
\end{proof}

\begin{corollary}[Simplex liftings]
\label{cor:simplex-lifting}
Let \(\mu\ge2\), and let \(C\subseteq\F_q^R\) be the \(q\)-ary simplex code
with parameters
\[
\left[\frac{q^\mu-1}{q-1},\ \mu,\ q^{\mu-1}\right]_q.
\]
Assume \(R=\sum_ir_i\) with \(r_i\le\min(n_i,m_i)\) for all \(i\), and choose
full-column-rank matrices \(V_i\in\F_q^{n_i\times r_i}\) and
\(W_i\in\F_q^{m_i\times r_i}\). Then the lifted code is
block-tensor-rank-extremal with
\[
\dim\cC=\mu,\qquad
\dSR(\cC)=q^{\mu-1},\qquad
\btrk(\cC)=\frac{q^\mu-1}{q-1}.
\]
The gap from the BTR threshold is
\[
\frac{q^\mu-1}{q-1}-(\mu+q^{\mu-1}-1)
=
\frac{q^{\mu-1}+q-2-(q-1)\mu}{q-1}.
\]
\end{corollary}

\begin{example}
\label{ex:simplex-lifting}
Take \(q=2\) and \(\mu=4\). The binary simplex code has parameters
\([15,4,8]_2\), and \(\Gri(4,8)=8+4+2+1=15\). Choose \(t=5\), \(r_i=3\) for
every \(i\), \(\bn=\bm=(3,3,3,3,3)\), and \(V_i=W_i=I_3\). The lifted code has
\[
\dim\cC=4,\qquad \dSR(\cC)=8,\qquad \btrk(\cC)=15.
\]
The BTR threshold is \(k+d-1=11\), and since \(N_2(4,8)\ge15\) by the
Griesmer bound, no linear \([11,4,8]_2\) block code exists. Thus, the lifted code
is \(N_q(k,d)\)-extremal but not BTR.
\end{example}

\begin{corollary}[MacDonald and Solomon--Stiffler liftings]
\label{cor:ss-macdonald-liftings}
Let \(C\subseteq\F_q^R\) be any Griesmer-optimal Solomon--Stiffler code
\cite{SolomonStiffler1965}, with parameters \([R,k,d]_q\) and
\(R=\Gri(k,d)\). Under the hypotheses of
Proposition~\ref{prop:griesmer-backbone-lifting}, the lifted sum-rank code is
block-tensor-rank-extremal. In particular, for the MacDonald subfamily with
parameters
\[
\left[\frac{q^\mu-q^s}{q-1},\ \mu,\ q^{\mu-1}-q^{s-1}\right]_q,\qquad
1\le s<\mu,
\]
the lifted code satisfies
\[
\dim\cC=\mu,\qquad
\dSR(\cC)=q^{\mu-1}-q^{s-1},\qquad
\btrk(\cC)=\frac{q^\mu-q^s}{q-1}.
\]
The gap from the BTR threshold is
\[
\frac{q^{\mu-1}-q^{s-1}}{q-1}-(\mu-1).
\]
\end{corollary}

The MDS and Griesmer constructions realize two different levels of extremality:
\[
\renewcommand{\arraystretch}{1.25}
\begin{array}{|c|c|c|c|}
\hline
\text{Backbone} & \text{Length} & \text{\(\btrk\) after lifting} & \text{Extremality}\\
\hline
\text{MDS / RS} & k+d-1 & k+d-1 & \text{BTR}\\
\hline
\text{Griesmer-optimal} & \Gri(k,d)=N_q(k,d) & N_q(k,d) & \text{\(N_q\)-extremal}\\
\hline
\end{array}
\]
When \(\Gri(k,d)>k+d-1\), the Griesmer bound rules out any linear
\([k+d-1,k,d]_q\) coordinate code. In that case BTR codes cannot exist for
those \((q,k,d)\), but Proposition~\ref{prop:griesmer-backbone-lifting}
still gives codes attaining the sharp lower bound \(\btrk=N_q(k,d)\).

\section{BTR versus MSRD, existing constructions, and complexity}
\label{sec:msrd-comparison}

BTR optimality and MSRD optimality measure different features of a sum-rank code: the former concerns the minimum number of block-simple tensors needed to
span the code, whereas the latter concerns whether the code dimension attains the sum-rank Singleton-type ceiling. We first show that, within our uniform
MDS-backbone lifting regime, BTR optimality cannot coexist with the uniform MSRD dimension. We then record a projection-gap criterion explaining how large
projected tensor ranks force a code away from the BTR threshold, and we apply this criterion to two recent sum-rank constructions in the literature. We close with a brief look at the encoding complexity of BTR codes obtained from extremal tuples.

\subsection{No simultaneous MSRD dimension in the uniform MDS-backbone construction}

\begin{proposition}
\label{prop:no-msrd-uniform-backbone}
Assume $t\ge 2$, $n\ge 2$, and $1\le d\le n\le m$. Under the parameter constraints of Corollary~\ref{cor:uniform-ngeqd} (uniform MDS-backbone regime with $tr=R$ and $n+m\ge r+d$), the construction cannot have the MSRD dimension $k=m(tn-d+1)$.
\end{proposition}

\begin{proof}
Suppose for contradiction that $k=m(tn-d+1)$. Then $R=k+d-1=m(tn-d+1)+d-1$ and $r=R/t$. The condition $n+m\ge r+d$ becomes
\[
n+m\ge\frac{m(tn-d+1)+d-1}{t}+d.
\]
Multiplying by $t$ and rearranging,
\begin{align*}
tn+tm
&\ge mtn-md+m+d-1+td\\
0
&\ge m(t(n-1)-d+1)+(d-1)-t(n-d).
\end{align*}
Equivalently,
\[
m\bigl(t(n-1)-d+1\bigr)\le t(n-d)-d+1.
\]
Set $A:=t(n-1)-d+1$ and $B:=t(n-d)-d+1$, so the inequality is $mA\le B$.

\emph{Claim 1: $A>0$.} We have $A=t(n-1)-(d-1)$. Since $t\ge 2$ and $n\ge 2$, $t(n-1)\ge 2$, while $d\le n$ gives $d-1\le n-1$, so
\[
A\ge t(n-1)-(n-1)=(t-1)(n-1)\ge 1\cdot 1=1>0.
\]

\emph{Claim 2: $A\ge B$, with equality iff $d=1$.} Indeed, $A-B=t(d-1)\ge 0$.

\emph{Final contradiction.} If $d=1$, then $A=B$, so $mA\le B=A$ forces $m\le 1$, contradicting $m\ge 2$. If $d\ge 2$, then $B<A$: either $B\ge 0$, giving $m\le B/A<1$, contradicting $m\ge 2$; or $B<0$, giving $mA\le B<0$, but $mA>0$. In all cases we have a contradiction.
\end{proof}

\begin{remark}
\label{rem:trade-off}
Proposition~\ref{prop:no-msrd-uniform-backbone} is not a general impossibility result for simultaneous MSRD and BTR optimality. Rather, it shows that the
specific uniform MDS-backbone lifting regime considered in Corollary~\ref{cor:uniform-ngeqd} lies in a dimension range incompatible with the uniform MSRD Singleton dimension. Thus, within this construction, BTR optimality is obtained at the cost of not reaching the MSRD dimension.
\end{remark}

\begin{proposition}[Projection-gap criterion]
\label{prop:projection-gap-criterion}
Let \(\mathcal C\) be a nonzero \([\mathbf n\times\mathbf m,k,d]_q\)
sum-rank code. For each nonzero projection
\(\mathcal C_i:=\pi_i(\mathcal C)\), write $ k_i:=\dim_{\mathbb F_q}\mathcal C_i, \, \delta_i:=d_{\rm rk}(\mathcal C_i).$ Then
\[
    \btrk(\mathcal C)-(k+d-1)
    \ge
    \sum_{i:\mathcal C_i\ne0} N_q(k_i,\delta_i)
    -(k+d-1).
\]
In particular, if $\sum_{i:\mathcal C_i\ne0} N_q(k_i,\delta_i)>k+d-1,$ then \(\mathcal C\) is not BTR, and the difference above gives a lower bound
on its gap from the BTR threshold. The weaker but explicit Singleton-level
version is
\[
    \btrk(\mathcal C)-(k+d-1)
    \ge
    \sum_{i:\mathcal C_i\ne0}(k_i+\delta_i-1)
    -(k+d-1).
\]
\end{proposition}

\begin{proof}
By the blockwise decomposition theorem and the rank-metric coordinate-code
bound applied to each nonzero projection,
\[
    \btrk(\mathcal C)
    =
    \sum_{i=1}^t \trk(\mathcal C_i)
    \ge
    \sum_{i:\mathcal C_i\ne0} N_q(k_i,\delta_i).
\]
Subtracting \(k+d-1\) gives the first inequality. The last inequality follows
from the Singleton bound \(N_q(k_i,\delta_i)\ge k_i+\delta_i-1\) for each
nonzero projection.
\end{proof}

\subsection{Application to the sum-rank codes from~\cite{Chen2023Good}}

We now apply the blockwise decomposition to the construction \(SR(C_0,\ldots,C_v)\) of Chen~\cite{Chen2023Good}. Let \(C_s\subseteq\F_{q^n}^t\), \(s=0,\ldots,v\), be \(\F_{q^n}\)-linear Hamming-metric codes. After fixing an \(\F_q\)-basis of \(\F_{q^n}\), each \(q\)-polynomial \(f(x)=a_0x+a_1x^q+\cdots+a_vx^{q^v}\) is identified with the matrix of the induced \(\F_q\)-linear map \(\F_{q^n}\to\F_{q^n}\); here \(x\) denotes the input variable in \(\F_{q^n}\). For coefficient vectors \((a_{s,1},\ldots,a_{s,t})\in C_s\), \(s=0,\ldots,v\), Chen's construction gives a sum-rank code
\[
SR(C_0,\ldots,C_v)\subseteq \cM((n,\ldots,n),(n,\ldots,n))=\bigoplus_{j=1}^t\F_q^{n\times n},
\]
whose \(j\)-th block is the matrix of \(a_{0,j}x+a_{1,j}x^q+\cdots+a_{v,j}x^{q^v}\). For each \(j\in[t]\), set
\[
S_j:=\{s\in\{0,\ldots,v\}:\pi_j(C_s)\ne 0\},
\]
where \(\pi_j:\F_{q^n}^t\to\F_{q^n}\) is the \(j\)-th coordinate projection. Then the \(j\)-th block projection is
\[
D_j:=\left\{\sum_{s\in S_j}a_sx^{q^s}:a_s\in\pi_j(C_s)\right\}\subseteq\F_q^{n\times n}.
\]
By Proposition~\ref{prop:blockwise-decomposition},
\[
\btrk(SR(C_0,\ldots,C_v))=\sum_{j=1}^t\trk(D_j).
\]

\textbf{Full-projection case.}
Assume that the coefficient codes have full support, so that \(\pi_j(C_s)=\F_{q^n}\) for every \(j\in[t]\) and \(s=0,\ldots,v\). Then \(S_j=\{0,\ldots,v\}\) for all \(j\), and all projected spaces \(D_j\) are equal to
\[
D=\left\{a_0x+a_1x^q+\cdots+a_vx^{q^v}:a_0,\ldots,a_v\in\F_{q^n}\right\}\subseteq\F_q^{n\times n}.
\]
If \(v<n\), then \(x,x^q,\ldots,x^{q^v}\) are linearly independent over \(\F_{q^n}\), and hence \(\dim_{\F_q}D=n(v+1)\). Therefore \(\btrk(SR(C_0,\ldots,C_v))=t\cdot\trk(D)\). In particular, for \(n=2\) and \(v=1\), we have \(\dim_{\F_q}D=4=n^2\), so \(D=\F_q^{2\times2}\). Hence \(\trk(D)=4\), the matrix units give the upper bound, while the dimension of \(D\) gives the matching lower bound. Thus, \(\btrk(SR(C_0,C_1))=4t\) whenever both coefficient codes have full support.
\par
\textbf{Comparison with \(k+d-1\).} Let \(n=2\) and \(v=1\). Take \(C_0=\F_{q^2}^t\), and let \(C_1\subseteq\F_{q^2}^t\) be an MDS code of Hamming distance \(w_1=d_H(C_1)=2\). Then \(\dim_{\F_{q^2}}C_1=t-1\), while \(C_0\) has Hamming distance \(w_0=d_H(C_0)=1\). Therefore
\[
k=\dim_{\F_q}SR(C_0,C_1)=2t+2(t-1)=4t-2.
\]
Chen's distance estimate gives \(d_{SR}\ge \min\{w_0n,w_1(n-1)\}=\min\{2,2\}=2\). This bound is attained by choosing a coefficient vector in \(C_0\) supported on one coordinate and setting all other coefficient layers equal to zero. Hence \(d_{SR}=2\). Consequently, \(k+d_{SR}-1=(4t-2)+2-1=4t-1\), whereas the full-projection computation above gives \(\btrk(SR(C_0,C_1))=4t\). Thus
\[
\btrk(SR(C_0,C_1))-(k+d_{SR}-1)=1,
\]
so this family is close to BTR but is not BTR. This example shows that the projection gap can be small: MSRD-related or sum-rank constructions need not be
far from the BTR threshold.


\subsection{Application to the MSRD codes from~\cite{LaoCheeChenVu2024MSRD}}

We next apply the blockwise decomposition to Construction~A of Lao et al.~\cite{LaoCheeChenVu2024MSRD}. We only recall the features of the construction needed for our block-tensor-rank analysis. Let \(\mathcal C\subseteq \mathcal M(\mathbf n,\mathbf m)\) be the resulting sum-rank code, generated by \(k\) generator codewords \(G_h=(G_{h,1},\ldots,G_{h,t})\), \(h\in[k]\), where \(G_{h,i}\in\F_q^{n_i\times m_i}\). Thus, the \(i\)-th block projection is
\[
\pi_i(\mathcal C)=\Span_{\F_q}\{G_{1,i},\ldots,G_{k,i}\}\subseteq\F_q^{n_i\times m_i}.
\]
Construction~A is built column by column from rank-metric building blocks. For \(1\le i\le J-1\), the matrices \(G_{h,i}\) are chosen from a rank-metric code in \(\F_q^{n_i\times m_i}\) of minimum rank distance \(n_i\). For \(i=J\), they are chosen from a nested rank-metric structure that guarantees a projected rank distance denoted by \(\Delta_J\). For \(J+1\le i\le t\), the projection is the full matrix space \(\F_q^{n_i\times m_i}\).


Write \(k_i:=\dim_{\F_q}\pi_i(\mathcal C)\). For \(1\le i\le J-1\), the projected rank distance is at least \(n_i\); for \(i=J\), it is at least
\(\Delta_J\); and for \(J+1\le i\le t\), the projection is the full matrix space \(\F_q^{n_i\times m_i}\). Applying Proposition~\ref{prop:projection-gap-criterion} to Construction~A gives
\[
\btrk(\mathcal C)\ge
\sum_{i=1}^{J-1}N_q(k_i,n_i)+N_q(k_J,\Delta_J)+\sum_{i=J+1}^{t} n_im_i.
\]
Indeed, for \(i>J\), the projection is the full matrix space, whose tensor rank is \(n_im_i\): the matrix units give the upper bound and the ambient dimension gives the matching lower bound.

\begin{example}
\label{ex:lao-numerical}
Consider Example~1 of~\cite{LaoCheeChenVu2024MSRD}, with ambient space
\[
\F_q^{6\times19}\oplus \F_q^{6\times18}\oplus \F_q^{6\times6}
\oplus \F_q^{2\times2}\oplus \F_q^{1\times2}.
\]
Here \(J=3\), \(k=18\), and the resulting MSRD code has minimum sum-rank distance \(d=17\). The projected rank-metric structure gives \((k_1,\delta_1)=(18,6)\), \((k_2,\delta_2)=(18,6)\), and \((k_3,\delta_3)=(18,4)\), while the fourth and fifth projections are full matrix spaces of dimensions \(4\) and \(2\), respectively. Hence
\[
\btrk(\mathcal C)\ge N_q(18,6)+N_q(18,6)+N_q(18,4)+4+2.
\]
Using the Singleton lower bound \(N_q(k,d)\ge k+d-1\), we obtain
\[
\btrk(\mathcal C)\ge 23+23+21+4+2=73.
\]
On the other hand, \(k+d-1=18+17-1=34\). Thus
\[
\btrk(\mathcal C)-(k+d-1)\ge 73-34=39.
\]
Thus this particular MSRD code is far from the BTR threshold. The large gap is explained by Proposition~\ref{prop:projection-gap-criterion}: the first three rank-metric projections already force large tensor-rank contributions, while the last two full projections contribute their full ambient dimensions. This shows that MSRD optimality alone does not control block tensor rank.
\end{example}

\subsection{Encoding complexity}
\label{sec:complexity}

Finally, we compare the cost of the coordinate-map encoder associated with an extremal tuple to a flat generator-matrix encoder. We count \(\F_q\)-multiplications and \(\F_q\)-additions in a dense implementation, without exploiting possible zero message-dependent coefficients or sparsity in the matrices \(V_i,W_i\). Storage refers to the size of the encoder representation, measured in \(\F_q\)-symbols.

\textbf{Generator-matrix encoding.} Identify \(\cM(\bn,\bm)\cong\F_q^{\sum_i n_im_i}\) via blockwise vectorisation. A \(k\times\sum_i n_im_i\) generator matrix \(G\) in systematic form \(G=[I_k\mid G']\) encodes \(a\in\F_q^k\) as \(aG=(a,aG')\). The non-systematic part has size \(k\times(\sum_i n_im_i-k)\), so computing \(aG'\) uses \(k(\sum_i n_im_i-k)\) multiplications and \((k-1)(\sum_i n_im_i-k)\) additions. The required storage is \(k(\sum_i n_im_i-k)\) symbols.

\textbf{Coordinate-map encoding.}
Suppose \(\cC=\phi^\br_{\bV,\bW}(C)\) for an MDS \([R,k,d]_q\) block code \(C\subseteq\F_q^R\) with systematic generator matrix \(U=[I_k\mid U']\in\F_q^{k\times R}\), where \(U'\in\F_q^{k\times(R-k)}\). Encoding factors as
\[
a\;\xrightarrow{\cdot U}\;c=aU\in\F_q^R\;\xrightarrow{\phi^\br_{\bV,\bW}}\;\phi^\br_{\bV,\bW}(c)\in\cM(\bn,\bm).
\]

\textbf{Stage~1.}
Computing \(c=aU\) requires only the non-systematic part \(aU'\), since the first \(k\) coordinates of \(c\) are \(a\). Hence Stage~1 uses \(k(R-k)\) multiplications and \((k-1)(R-k)\) additions.

\textbf{Stage~2.}
For each block \(i\), compute \(V_i\diag(c^{(i)})W_i^\top\). When \(r_i\ge1\), this can be done by first scaling the \(r_i\) columns of \(V_i\) by the entries of \(c^{(i)}\), costing \(n_ir_i\) multiplications, and then multiplying the resulting \(n_i\times r_i\) matrix by \(W_i^\top\), costing \(n_ir_im_i\) multiplications and \(n_i(r_i-1)m_i\) additions. Blocks with \(r_i=0\) contribute no arithmetic cost. Thus, Stage~2 uses
\[
\sum_{i=1}^t n_ir_i(m_i+1)
\]
multiplications and
\[
\sum_{i=1}^t n_i\max\{r_i-1,0\}m_i
\]
additions. The storage required for the lifting data \((\bV,\bW)\) is \(\sum_{i=1}^t r_i(n_i+m_i)\) symbols, in addition to the \(k(R-k)\) symbols needed to store \(U'\).

\begin{proposition}
\label{prop:complexity}
Let \(\cC\) be an extremal BTR \([\bn\times\bm,k,d]_q\) code arising from an MDS block code with \(R=k+d-1\). Then:
\begin{itemize}
\item The coordinate-map encoder uses less storage than a systematic-form generator matrix if and only if
\[
kR+\sum_{i=1}^t r_i(n_i+m_i)<k\sum_{i=1}^t n_im_i.
\]
\item The multiplication cost is \(k(R-k)+\sum_{i=1}^t n_ir_i(m_i+1)\), versus \(k(\sum_i n_im_i-k)\) for the generator-matrix encoder.
\item The addition cost is \((k-1)(R-k)+\sum_{i=1}^t n_i\max\{r_i-1,0\}m_i\).
\end{itemize}
\end{proposition}

\begin{proof}
The systematic generator matrix uses \(k(\sum_i n_im_i-k)\) symbols, while the coordinate-map encoder uses \(k(R-k)+\sum_i r_i(n_i+m_i)\) symbols. Hence the coordinate-map encoder uses less storage exactly when
\[
k(R-k)+\sum_i r_i(n_i+m_i)<k\Bigl(\sum_i n_im_i-k\Bigr),
\]
which is equivalent to \(kR+\sum_i r_i(n_i+m_i)<k\sum_i n_im_i\). The arithmetic costs are obtained by adding the Stage~1 and Stage~2 counts above.
\end{proof}

\begin{remark}
\label{rem:complexity}
When \(t=1\), Stage~2 reduces to a single product \(V\diag(c)W^\top\), recovering the rank-metric coordinate map of~\cite[Section 4.2]{BNRS19}. For \(t\ge2\), whether the coordinate-map encoder beats the generator-matrix encoder depends on \((\br,\bn,\bm)\); the Stage~2 cost can grow faster than \(\sum_i n_im_i\) when some \(r_i\) is large relative to its block size.
\end{remark}

\section{Conclusion and future work}
\label{sec:conclusion}
We introduced and studied a block-tensor-rank invariant for sum-rank metric codes through block-simple tensors and the associated notion of block tensor rank.
The main structural result is a blockwise decomposition theorem, showing that the block tensor rank of a sum-rank code is the sum of the tensor ranks of its block projections. This identity connects the sum-rank setting directly to the tensor-rank theory of rank-metric codes.

We also established projection-wise bound and coordinate-code bound for block tensor rank. The coordinate-code bound is obtained through a coordinate-code argument and yields a Singleton-type inequality, with a Griesmer refinement. These bounds naturally lead to two optimality notions: BTR codes, which meet the Singleton-level lower bound, and block-tensor-rank-extremal codes, which meet the sharper Hamming-code lower bound. On the constructive side, we obtained BTR families by lifting MDS backbones, and block-tensor-rank-extremal families by lifting Griesmer-optimal backbones. We also compared BTR optimality with MSRD optimality and showed that these two notions capture different structural features of sum-rank codes.

Several questions remain open. One direction is to classify parameter regimes in which BTR codes exist, especially beyond the lifting regimes considered here. Another is to sharpen the projection-wise lower bounds by using finer tensor-rank information for the individual rank-metric projections. It would also be interesting to understand when MSRD codes can be close to BTR, and when their block tensor rank must be much larger than the Singleton-level threshold.
\par
\textbf{Possible cryptographic connections.}
Code-based cryptography is one of the main approaches to post-quantum cryptography, and rank-metric codes have been studied as compact algebraic alternatives to Hamming-metric codes~\cite{Loidreau2010RankMetricMcEliece}. More recently, sum-rank-metric codes have appeared in cryptographic contexts through generic decoding problems and structural questions in the sum-rank metric~\cite{PuchingerRennerRosenkilde2022,SantonastasoZullo2025}.

The block tensor rank developed in this paper suggests a possible structural viewpoint for the study of sum-rank code families in this context. It measures how economically a sum-rank code can be spanned by rank-one components supported inside individual blocks, while the blockwise decomposition relates this invariant to the tensor ranks of the rank-metric projections. Thus, small block tensor rank may indicate additional algebraic structure, whereas large projected tensor ranks may serve as evidence of more complicated internal rank-metric behavior.

This leads to several natural questions. Can block tensor rank be used as a distinguisher between random sum-rank codes and highly structured code families? Can unusually small block tensor rank reveal information relevant to structural decoding attacks? Conversely, can one design efficiently decodable sum-rank code families whose block tensor rank and projected tensor ranks are provably large? These questions suggest a possible direction for future work at the interface of sum-rank coding theory, tensor methods, and post-quantum cryptography.

\section*{Acknowledgements}
The work of Hoang Ta was funded by the Ministry of Education and Training of Vietnam under project code CT2025.EA.BKA.08.
The research of Huimin Lao is supported by the National Research Foundation, Singapore and Infocomm Media Development Authority under its Trust Tech Funding Initiative. 
Any opinions, findings and conclusions or recommendations expressed in this material are those of the author(s)
and do not reflect the views of National Research Foundation, Singapore and Infocomm Media Development Authority.


\bibliographystyle{unsrt}
\bibliography{references}

\end{document}